\documentclass[a4paper,UKenglish,cleveref, autoref]{lipics-v2019}
\usepackage[utf8]{inputenc}
\usepackage{amsmath,comment}
\usepackage{hyperref}
\usepackage{amsthm}
\usepackage{tikz}
\usepackage{slashbox}
\usepackage{boxedminipage}
\usepackage{xypic}

\usepackage{tikz}
\usetikzlibrary{decorations.pathreplacing}
\usetikzlibrary{shapes.geometric}
\tikzstyle{snode}=[circle,draw=black,fill=white,thick, inner sep=0pt ,minimum size=1.4mm]
\tikzstyle{bnode}=[circle ,draw=black,fill=black,thick, inner sep=0pt ,minimum size=1.4mm]

\newcommand{\NP}{{\sf NP}}

\DeclareMathOperator{\chair}{chair}

\theoremstyle{plan}
\newtheorem{observation}[theorem]{Observation}

\title{Partitioning $H$-Free Graphs of Bounded Diameter\footnote{An extended abstract of this paper will appear in the proceedings of ISAAC 2021.}} 
 
\titlerunning{Partitioning $H$-Free Graphs of Bounded Diameter}

\author{Christoph Brause}{Technische Universit\"at Bergakademie Freiberg}{brause@math.tu-freiberg.de}{}{}
\author{Petr Golovach}{University of Bergen}{petr.golovach@ii.uib.no}{https://orcid.org/0000-0002-2619-2990}{}
\author{Barnaby Martin}{Department of Computer Science, Durham University, Durham, United Kingdom}{barnaby.d.martin@durham.ac.uk}{}{}
\author{Dani\"el Paulusma}{Department of Computer Science, Durham University, Durham United Kingdom}{daniel.paulusma@durham.ac.uk}{0000-0001-5945-9287}{Supported by the Leverhulme Trust (RPG-2016-258).}
\author{Siani Smith}{Department of Computer Science, Durham University, Durham, United Kingdom}{siani.smith@durham.ac.uk}{}{}

\authorrunning{C. Brause, P.A. Golovach, B. Martin, D. Paulusma and S. Smith }
\Copyright{Christoph Brause, Petr Golovach, Barnaby Martin, Dani\"el Paulusma and Siani Smith}

\ccsdesc[500]{Mathematics of computing~Graph theory}

\keywords{vertex partitioning problem, $H$-free, diameter, complexity dichotomy}

\nolinenumbers 
\hideLIPIcs

\begin{document}

\maketitle

\begin{abstract}
A natural way of increasing our understanding of \NP-complete graph problems is to restrict the input to a special graph class. Classes of $H$-free graphs, that is, graphs that do not contain some graph~$H$ as an induced subgraph, have proven to be an ideal testbed for such a complexity study. However, if the forbidden graph $H$ contains a cycle or claw, then these problems often stay \NP-complete. A recent complexity study (MFCS 2019) on the {\sc $k$-Colouring} problem shows that we may still obtain tractable results if we also bound the diameter of the $H$-free input graph. We continue this line of research by initiating a complexity study on the impact of bounding the diameter for a variety of classical vertex partitioning problems restricted to $H$-free graphs.  We prove that bounding the diameter does not help for {\sc Independent Set}, but leads to new tractable cases for problems closely related to {\sc $3$-Colouring}. That is, we show that {\sc Near-Bipartiteness}, {\sc Independent Feedback Vertex Set}, {\sc Independent Odd Cycle Transversal}, {\sc Acyclic $3$-Colouring} and {\sc Star $3$-Colouring} are all polynomial-time solvable for chair-free graphs of bounded diameter. To obtain these results we exploit a new structural property of $3$-colourable chair-free graphs.
\end{abstract}

\section{Introduction}\label{s-intro}

Many well-known graph problems are \NP-complete in general but become polynomial-time solvable under input restrictions. We focus on problems that partition the vertex set $V$ of a graph $G$ into sets $V_1,\ldots,V_k$ such that each $V_i$ satisfies some property $\pi_i$ and where $V_1$ might have the extra condition of being large or being small. For instance, the {\sc $k$-Colouring} problem is to decide if $V$ can be partitioned into sets $V_1,\ldots,V_k$, called {\it colour classes}, such that each $V_i$ is an independent set. To give another example, the {\sc Independent Set} problem is to decide if $V$ can be partitioned into sets $V_1$ and $V_2$ where $V_1$ is independent and $|V_1|\geq p$ for some given integer~$p$. Our underlying goal is to understand which graph properties ensure tractability of these problems and which properties cause the computational hardness. In the literature, input is restricted in various ways. In particular, {\it hereditary} graph classes have been considered. 

Hereditary graph classes are the classes of graphs closed under vertex deletion. They form a natural and rich framework that cover many well-known graph classes (see, for example,~\cite{BLS99}). Moreover, they enable a systematic study on the computational complexity of graph problems under input restrictions. The reason is that a graph class~${\cal G}$ is hereditary if and only if ${\cal G}$ can be characterized by a set ${\cal F}_{\cal G}$ of forbidden induced subgraphs; we also say that ${\cal G}$ is {\it ${\cal F}_{\cal G}$-free}. A natural starting point for a systematic study is the case where ${\cal F}_{\cal G}$ has size~$1$, say ${\cal F}_{\cal G}=\{H\}$ for some graph $H$. In this case the graphs in ${\cal G}$ are said to be {\it $H$-free}. In other words, no graph $G\in {\cal G}$ can be modified into $H$ by a sequence of vertex deletions.

In the literature, there are extensive studies on $H$-free graphs; for example, on bull-free graphs~\cite{Ch12} and claw-free graphs~\cite{CS05,HMLW11}. There also exist several surveys on graph problems or graph parameters for hereditary graph classes that are characterized by a small set of forbidden induced subgraphs, for example, on {\sc Colouring}~\cite{GJPS17,RS04} and clique-width~\cite{DJP19}. 

A well-known dichotomy on {\sc Colouring} restricted to $H$-free graphs is due to Kr\'al', Kratochv\'{\i}l, Tuza, and Woeginger~\cite{KKTW01}. Namely, {\sc Colouring} on $H$-free graphs is polynomial-time solvable  if $H$ is an induced subgraph of 
$P_4$ (the $4$-vertex path) or of $P_1+P_3$ (the disjoint union of $P_1$ and $P_3$) and it is \NP-complete otherwise. Recently, similar but almost-complete dichotomies (up to one missing case each) were established for {\sc Acyclic Colouring}, {\sc Star Colouring} and {\sc Injective Colouring}~\cite{BJMOPS}.  In particular, all these problems stay \NP-complete if the forbidden induced subgraph $H$ has a cycle or claw (the $4$-vertex star~$K_{1,3}$). Moreover, the latter holds even if the number of colours $k$ is fixed, i.e., not part of the input. 

Several other vertex partitioning problems on $H$-free graphs stay \NP-complete as well if $H$ has a cycle or claw. Examples of such problems include {\sc (Independent) Feedback Vertex Set}~\cite{BDFJP19,Mu17,Po74}, 
{\sc (Independent) Odd Cycle Transversal}~\cite{BDFJP19,CHJMP18} and {\sc Even Cycle Transversal}~\cite{PPR}.
Hence, for all these problems, if $H$ is a cycle or claw, then we need to add more structure to the class of input graphs in order to find tractable results for $H$-free graphs. One way of doing this is to bound the {\it diameter} of the input graph $G$ for some problem. Our research question then becomes:

\medskip
\noindent
{\it Does bounding the diameter of an $H$-free graph lead to new tractability results?}

\medskip
\noindent
We note that graph classes of diameter at most~$d$ are hereditary if and only if $d\leq 1$. Many graph problems, such as {\sc Colouring},  {\sc Acyclic Colouring}, {\sc Star Colouring}, {\sc Clique} and {\sc Independent Set} stay \NP-complete even for graphs of diameter~$2$. The reason is that we can take an arbitrary instance $(G,k)$ from such a problem and add a dominating vertex to obtain a new graph $G'$: $(G,k)$ is a yes-instance if and only if $(G',k+1)$ is a yes-instance.

This approach of adding a dominating vertex does not work if we consider $3$-{\sc Colouring}. Mertzios and Spirakis~\cite{MS16} proved in a highly nontrivial way that {\sc $3$-Colouring} is \NP-complete even for triangle-free graphs of diameter at most~$3$. However, determining the complexity of {\sc $3$-Colouring} for graphs of diameter~$2$ is a notoriously open problem (see~\cite{BKM12,BFGP13,MPS19,MS16,Pa15}); we refer to~\cite{MS16} and~\cite{DPR} for subexponential-time algorithms for {\sc List $3$-Colouring} on graphs of diameter at most~$2$. It is also known that {\sc Acyclic $3$-Colouring} and {\sc Star Colouring}, restricted to graphs of diameter at most~$d$, are polynomial-time solvable if $d=2$ or $d=3$, respectively, but \NP-complete if $d=5$ or $d=8$, respectively~\cite{BGMPS21}. Moreover, the related problems {\sc Near Bipartiteness} and {\sc Independent Feedback Vertex Set}, which we define below, are polynomial-time solvable for~$d=2$ but \NP-complete for $d=3$~\cite{BDFJP18}.
These results also show that bounding the diameter on its own (without forbidding any induced graph $H$) does not suffice.

We refer to~\cite{MPS19,MPS21} for a number of results on {\sc $3$-Colouring} and {\sc List $3$-Colouring} for $H$-free graphs of bounded diameter, where $H$ is a cycle or a {\it polyad}, which is a tree where exactly one vertex has degree at least~$3$ (polyads are also known as {\it subdivided stars}).
 One crucial observation in~\cite{MPS19}, based on an application of Ramsey's Theorem,  was the starting point of this investigation: for all integers $d,k$, the {\sc $k$-Colouring} problem is constant-time solvable on claw-free graphs of diameter at most~$d$. In the same paper~\cite{MPS19}, this result was generalized to the case where $H$ is the {\it chair}, which is the graph obtained from the claw after subdividing exactly one of its edges (see also Figure~\ref{f-st}). The chair is also known as the {\it fork}.

 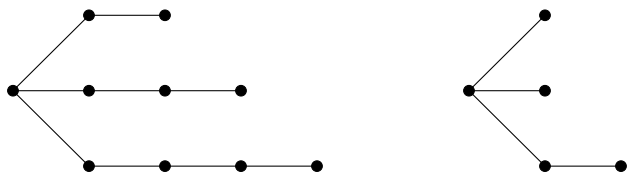
\begin{figure}
  \hspace*{2cm}
\begin{tikzpicture}[scale=1]
\draw[fill=black] (4,0) circle [radius=2pt] ;
\draw[fill=black] (5,0) circle [radius=2pt] ;
\draw[fill=black] (5,1) circle [radius=2pt];
\draw[fill=black] (5,-1) circle [radius=2pt];
\draw[fill=black] (6,-1) circle [radius=2pt];
\draw(4,0)--(5,0);
\draw(4,0)--(5,1);
\draw(4,0)--(5,-1)--(6,-1);   

\draw (0,1)--(-1,1)--(-2,0)--(-1,-1)--(2,-1) (-2,0)--(1,0); \draw[fill=black] (-1,1) circle [radius=2pt] (0,1) circle [radius=2pt] (-2,0) circle [radius=2pt] (-1,0) circle [radius=2pt] (0,0) circle [radius=2pt] (1,0) circle [radius=2pt] (-1,-1) circle [radius=2pt] (0,-1) circle [radius=2pt] (1,-1) circle [radius=2pt] (2,-1) circle [radius=2pt];

\end{tikzpicture}
\caption{Left: the graph $S_{2,3,4}$. Right: the graph $S_{1,1,2}$ also known as the chair or fork.}\label{f-st}
\end{figure}

\subsection*{Our Results} 
We first consider the {\sc Independent Set} problem. For this problem we will prove new \NP-hardness results that show that bounding the diameter does not help. To explain this, let the graph~$S_{h,i,j}$, for $1\leq h\leq i\leq j$, be the \emph{subdivided claw}, which is the tree with one vertex~$x$ of degree~$3$ and exactly three leaves, which are at distance~$h$,~$i$ and~$j$ from~$x$, respectively; see \figurename~\ref{f-st} for an example.  
Note that $S_{1,1,1}$ is the claw~$K_{1,3}$, the graph $S_{1,1,2}$ is the chair and that every subdivided claw is a polyad.
Let ${\cal S}$ be the set of graphs, each connected component of which is a subdivided claw or a path. 
Alekseev~\cite{Al82} proved that for every finite set of graphs~${\cal F}$, if no graph from ${\cal F}$ belongs to ${\cal S}$, then {\sc Independent Set} is \NP-complete for the class of ${\cal F}$-free graphs. In Section~\ref{s-hard}, we show that exactly the same \NP-completeness result holds for the class of ${\cal F}$-free graphs {\it of diameter~$2$} if and only if $|{\cal F}|=1$. 

We then turn to the class of $H$-free graphs where $H$ is a polyad. First, we focus on the case where $H$ is the chair. In Section~\ref{s-h},  we prove
 that for every integer $d\geq 1$, a number of vertex partitioning problems that require yes-instances to be $3$-colourable become polynomial-time solvable on chair-free graphs of diameter at most~$d$. The problems are {\sc Acyclic $3$-Colouring}, {\sc Star $3$-Colouring}, {\sc Near-Bipartiteness}, {\sc Independent Feedback Vertex Set} and {\sc Independent Odd Cycle Transversal}. We define these problems below. 
 
 Our proof is based on a common strategy. Namely, we determine the following for every chair-free $3$-colourable non-bipartite input graph $G$ of bounded diameter: either $G$ has a constant number of $3$-colourings or there exists a set $S$ such that $G-S$ has this property. We prove that we can let $S$ be the set of private neighbours of some vertex $u$ of a triangle on vertices $u,v,w$, that is, the vertices of $S$ are adjacent to neither $v$ nor $w$. We then consider each constructed $3$-colouring $c$ and determine in polynomial time if we can extend $c$ to a solution for the vertex partitioning problem under consideration.

In Section~\ref{s-limitations} we prove that there is little hope of a full extension from the chair to arbitrary polyads $H$. To be more precise, we prove that for {\sc Acyclic $3$-Colouring}, {\sc Star $3$-Colouring} and {\sc Independent Odd Cycle Transversal}, there exists a polyad~$H$ and a constant~$d$ such that each of these problems is \NP-complete for the class of $H$-free graphs of diameter at most~$d$. 
We finish our paper with some relevant open problems in Section~\ref{s-con}.

\subsection*{Additional Terminology}
Let $G=(V,E)$ be a graph. A {\it colouring} of $G$ is a mapping $c:V\to \{1,2,\ldots\}$ with $c(u)\neq c(v)$ for every pair of adjacent vertices $u$ and $v$ in $G$. If $c(u)\in \{1,\ldots,k\}$ for every $u\in V$, then $c$ is also called a {\it $k$-colouring}. 
Note that the sets $\{u\in V(G)\; |\; c(u)=i\}$ for $i\in \{1,\ldots,k\}$ are the corresponding colour classes. If a $k$-colouring exists, then $G$ is said to be {\it $k$-colourable}.

A graph is {\it acyclic $3$-colourable} or {\it star $3$-colourable} if it is $3$-colourable and every two colour classes induce a forest or a star forest, respectively (in this context, the $P_1$ and $P_2$ are seen as stars). The corresponding decision problems are {\sc Acyclic $3$-Colouring} and {\sc Star $3$-Colouring}. A graph~$G$ is {\it near-bipartite} if its vertex set can be partitioned into an independent set $I$ and a forest $F$; we also say that $I$ is an {\it independent feedback vertex set} of~$G$. The problems {\sc Near-Bipartiteness} and {\sc Independent Feedback Vertex Set} are to decide if a graph is near-bipartite or has an independent feedback vertex set of size at most $k$ for some given integer $k$. A subset $S\subseteq V$ of a graph $G=(V,E)$ is an {\it independent odd cycle transversal} if $S$ is independent and $G-S$ is bipartite. Note that a graph is $3$-colourable if and only if it has an independent odd cycle transversal. The {\sc Independent Odd Cycle Transversal} problem is to decide if a given graph has an independent odd cycle transversal of size at most $k$ for some given integer $k$.

Let $C_r$, $P_r$ and $K_r$ be the cycle, path and complete graph on $r$ vertices. The graph $G+H=(V(G)\cup V(H),E(G)\cup E(H))$ is the disjoint union of graphs $G$ and $H$, and $sG$ is the disjoint union of $s$ copies of $G$. A graph~$G$ is ${\cal H}$-free if $G$ is $H$-free for every $H\in {\cal H}$.

Let $G$ be a graph. For a vertex $u\in V(G)$, we let $N(u)=\{v\; |\; uv\in E(G)\}$ denote the {\it neighbourhood} of $u$ in $G$. For a subset $U\subseteq V(G)$, we write $N(U)=\bigcup_{u\in U}N(u)\setminus U$.

\section{Independent Set}\label{s-hard}

We let ${\cal S}$ denote the set of graphs, each connected component of which is either a subdivided claw or a path. The following well-known result is due to Alekseev.

\begin{theorem}[\cite{Al82}]\label{t-al82}
	Let ${\cal F}$ be a finite set of graphs. If no graph from ${\cal F}$ belongs to ${\cal S}$, then {\sc Independent Set} is \NP-complete for ${\cal F}$-free graphs.
\end{theorem}

We strengthen Theorem~\ref{t-al82} to ${\cal F}$-free graphs of diameter~$2$ if $|{\cal F}|=1$.
We need two lemmas.

\medskip
\begin{figure}[h]
\centering
\begin{tikzpicture}[scale=1.2]
\foreach \c in {0,2}
	\draw (0,1)--(\c,0);
\foreach \c in {0,3}
	\draw (1,1)--(\c,0);
\foreach \c in {1,3}
	\draw (2,1)--(\c,0);
\foreach \c in {1,4}
	\draw (3,1)--(\c,0);
\foreach \c in {2,4}
	\draw (4,1)--(\c,0);

\draw[out=-60, in=-120] (0,0) to (4,0);
\draw (0,0) node[bnode]{}--(1,0) node[bnode]{}--(2,0) node[bnode]{}--(3,0) node[bnode]{}--(4,0) node[bnode]{};

\foreach \z in  {-.5,0.5,4.5}
\foreach \x in  {0,1,2,3,4}
	\draw (\z,2.3)node[snode]{}--(\x,1)node[snode]{};

\draw (-.2,0.05) node[below]{\small $u_1$};	
\draw (1,0) node[below]{\small $u_2$};	
\draw (2,0) node[below]{\small $u_3$};	
\draw (3,0) node[below]{\small $u_4$};
\draw (4.25,0.05) node[below]{\small $u_5$};	
\draw (-.5,2.3) node[above]{\small $y_1$};	
\draw (0.5,2.3) node[above]{\small $y_2$};
\draw (2.5,2.3) node[above]{\small \ldots};
\draw (4.5,2.3) node[above]{\small $y_{25}$};	
\end{tikzpicture}
\caption{The graph $G'$ in the proof of Lemma \ref{l-ic3d2} if $G=C_5$. The black vertices are those of $G$ and the white vertices in the middle are $x_{u_1u_3},x_{u_1u_4},x_{u_2u_4},x_{u_2u_5},$ and $x_{u_3u_5}$ from left to right.}\label{f-p1p2}
\end{figure}

\begin{lemma}\label{l-ic3d2}
	For every $r\geq 3$, {\sc Independent Set} is \NP-complete for $C_r$-free graphs of diameter~$2$.
\end{lemma}

\begin{proof}
	First suppose that $r=3$. By Theorem~\ref{t-al82}, {\sc Independent Set} is \NP-complete for $C_3$-free graphs.
	Let $(G,k)$ be an instance of {\sc Independent Set}, where $G$ is an $n$-vertex $C_3$-free graph. We may assume without loss of generality that $n\geq 2$ and $k\geq 2$; otherwise, the problem is trivial. We also assume that $G$ has no dominating vertex; else, if $u$ is a dominating vertex of $G$, then $G$ is a star with centre $u$ and the leaves of $G$ form a maximum independent set of size $n-1$. From $G$ we now construct a graph $G'$ as follows; see also Figure~\ref{f-p1p2}.

	\begin{itemize}
		\item Construct a copy of $G$.
		\item For every pair $\{u,v\}$ of nonadjacent vertices of $G$, 
		\begin{itemize}
			\item construct an arbitrary inclusion maximal independent set $I_{uv}$ of $G$ containing $u$ and $v$,
			\item construct a vertex $x_{uv}$ and make $x_{uv}$ adjacent to every vertex of $I_{uv}$.
		\end{itemize}
		Denote by $X$ the set of all vertices $x_{uv}$ constructed for the pairs of nonadjacent $u$ and $v$.
		\item Construct an
		independent set of $n^2$ vertices $Y$ and make every vertex of $Y$ adjacent to every vertex of~$X$.  
	\end{itemize}      
	
\noindent
	Observe that for every pair $\{u,v\}$ of nonadjacent vertices of $G$, $I_{uv}$ can be constructed by the straightforward greedy procedure starting from the set $\{u,v\}$. This implies that $G$ can be constructed in polynomial time.
	Our construction immediately implies that $G'$ is triangle free, because the initial graph $G$ is triangle-free and the neighbourhood of every vertex from $X\cup Y$ is an  independent set. We claim that the diameter of $G'$ is at most~$2$.
	
	Consider two distinct vertices $u,v\in V(G')$. We show that either $u$ and $v$ are adjacent or they have a common neighbour $z$. This is easy to see if $u,v\in X\cup Y$. Suppose that $u\in V(G)$ and $v$ is not adjacent to $u$. 
	If $v\in V(G)$, then $G'$ contains $x_{uv}$ that is adjacent to both $u$ and $v$ as required. Let $v\in X$. Then $v=x_{u'v'}$ is constructed for some nonadjacent $u',v'\in V(G)$.  We have that $u\notin I_{u'v'}$, because $u$ and $v$ are not adjacent. We also have that $I_{u'v'}\cup \{u\}$ is not an independent set, because $I_{u'v'}$ is inclusion maximal. Hence, $u$ is adjacent to some $z\in I_{u'v'}$ which is a common neighbour of $u$ and $v$.
	Assume that $v\in Y$. Since $G'$ has no universal vertex, $u$ is not adjacent to some $v'\in V(G)$. Therefore, there is $x_{uv'}\in X$. Clearly, this vertex is adjacent to $u$ and $v$. This completes the case analysis. Since $u$ and $v$ are arbitrary, we conclude that the diameter of $G'$ is at most two.
	
	Let $k'=n^2+k$. We show that $G$ has an independent set of size at least $k$ if and only if $G'$ has an independent set of size at least $k'$. 
	
	Assume that $G$ has an independent set $U$ of size at least $k$. Consider this set in the copy of $G$ in $G'$ and set $W=U\cup Y$. Clearly, $W$ is an independent set of $G'$ and $|W|=|U|+|Y|\geq n^2+k=k'$. For the opposite direction, let $W$ be an independent set of $G'$ of size at least $k'$. Observe that $W\cap Y\neq \emptyset$. Otherwise, $W\subseteq V(G)\cup X$ and $|W|\leq |V(G)|+|X|\leq n+\binom{n}{2}<n^2$, because $n\geq 2$. Let $y\in W\cap Y$. Because $N_{G'}(y)=X$, $W\cap X=\emptyset$. Let $U=W\cap V(G)$. We obtain that $|U|=|W|-|W\cap Y|\geq |W|-n^2\geq k'-n^2=k$. Hence, $U$ is an independent set of $G$ of size at least $k$.
	
	\medskip
	\noindent
	Now suppose that $r\geq 4$. By Theorem~\ref{t-al82}, {\sc Independent Set} is \NP-complete for $C_r$-free graphs. Let $(G,k)$ be an instance of {\sc Independent Set}, where $G$ is a $C_r$-free graph. We add a dominating vertex $u$ to $G$. This yields a graph $G'$, which is also $C_r$-free but which has diameter~$2$. Moreover, $G$ has an independent set of size at least $k$ if and only if $G'$ has an independent set of size at least $k$. This completes the proof.
\end{proof}

\begin{lemma}\label{l-k14}
	{\sc Independent Set} is \NP-complete for $K_{1,4}$-free graphs of diameter~$2$.
\end{lemma}

\begin{proof}
	By Theorem~\ref{t-al82}, {\sc Independent Set} is \NP-hard for $(C_3,K_{1,4})$-free graphs. Let $(G,k)$ be an instance of {\sc Independent Set}, where $G$ is a $(C_3,K_{1,4})$-free graph of order $n$ and size $m$. We may assume without loss of generality that $G$ is connected. Note that $G$ is subcubic, that is, has maximum degree at most~$3$. We may assume without loss of generality that
	$G$ has no vertices of degree~$1$ (as we can pick such vertices to be in the independent set and remove their neighbours from $G$ until this operation can no longer be applied).
	To $G$, we add for every pair of edges $e_1$ and $e_2$ of $G$ that do not share an end-vertex,
	a new vertex~$x_{e_1,e_2}$ and an edge between $x_{e_1,e_2}$ and the two end-vertices of both $e_1$ and $e_2$. 
We also add a new vertex $y$ and all edges such that 
	$\{x_{e_1,e_2}:e_1,e_2\}\cup \{y\}$ induces a complete subgraph. Let $G'$ be the resulting graph and let $X:= \{x_{e_1,e_2}:e_1,e_2\}\cup \{y\}$. See also Figure~\ref{f-p5}.
\begin{figure}[h]
\centering
\begin{tikzpicture}[scale=1.2]
\foreach \c in {0,1,2,3}
	\draw (0,1)--(\c,0);
\foreach \c in {0,1,2,4}
	\draw (1,1)--(\c,0);
\foreach \c in {0,1,3,4}
	\draw (2,1)--(\c,0);
\foreach \c in {0,2,3,4}
	\draw (3,1)--(\c,0);
\foreach \c in {1,2,3,4}
	\draw (4,1)--(\c,0);

\draw[out=-60, in=-120] (0,0) to (4,0);
\draw (0,0) node[bnode]{}--(1,0) node[bnode]{}--(2,0) node[bnode]{}--(3,0) node[bnode]{}--(4,0) node[bnode]{};

\draw[bend left] (0,1) to (4,1);
\draw[bend left] (0,1) to (3,1);
\draw[bend left] (0,1) to (2,1);
\draw[bend left] (1,1) to (4,1);
\draw[bend left] (1,1) to (3,1);
\draw[bend left] (2,1) to (4,1);
\draw (0,1) --(4,1);

\foreach \x in{0,1,2,3,4}
	\draw (2,2.3)node[snode]{}--(\x,1)node[snode]{};

\draw (0.5,0.1) node[below]{\small $e_1$};	
\draw (1.5,0.1) node[below]{\small $e_2$};	
\draw (2.5,0.1) node[below]{\small $e_3$};	
\draw (3.5,0.1) node[below]{\small $e_4$};	
\draw (2,-0.9) node[below]{\small $e_5$};	
\draw (2,2.3) node[above]{\small $y$};	

\end{tikzpicture}
\caption{The graph $G'$ in the proof of Lemma \ref{l-k14} if $G=C_5$. The black vertices are those of $G$ and the white vertices in the middle are $x_{e_1e_3},x_{e_2e_5},x_{e_1e_4},x_{e_3e_5},$ and $x_{e_2e_4}$ from left to right.}\label{f-p5}
\end{figure}
	
We claim that $G'$ is $K_{1,4}$-free. For a contradiction, assume that there is some induced $K_{1,4}$ in $G'$, say $z$ is its centre vertex and $z_1,z_2,z_3,z_4$ are its leafs. Since $N_{G'}(x)$ can be partitioned into at most three cliques for each $x\in X$, we have $z\in V(G)$. Furthermore, since $G$ is subcubic and $X$ is a clique, $X$ contains at least one and at most one vertex of $\{z_1,z_2,z_3,z_4\}$. As we see now, $N_G(z)\subseteq \{z_1,z_2,z_3,z_4\}$. However, each vertex of $X$ that is adjacent to $z$ has a neighbour in $N_G(z)$, which contradicts our assumption that $z$ is the centre vertex of an induced $K_{1,4}$ in $G'$.
	
To show that $G'$ has diameter~$2$, first consider an arbitrary vertex $u$ of $G$.
As $G$ is connected, there is a vertex $v\in N_G(u)$. As $G$ has minimum degree at least $2$ and $G$ is triangle-free, there are vertices $u'\in N_G(u)\setminus\{v\}$ and $v'\in N_G(v)\setminus\{u,u'\}$. For $e_1=uu'$ and $e_2=vv'$,  $e_1$ and $e_2$ do not have a common end-vertex, $u,x\in N_{G'}[x_{e_1,e_2}]$, and $dist_{G'}(u,x)\leq 2$ for each $x\in X$. Additionally, if $w\in V(G)\setminus\{u\}$ is a vertex with $dist_G(u,w)\geq 3$, then, since $G$ is connected, there are edges $e_3,e_4\in E(G)$ such that $u$ is incident to $e_3$, $w$ is incident to $e_4$, and the two edges $e_3,e_4$ have no common incident vertex. Thus, $u,w\in N_G(x_{e_3,e_4})$ and so $dist_{G'}(u,w)\leq 2$. As $X$ is a clique, we conclude that $G'$ has diameter $2$.
	
	We observe that $I\cup\{y\}$ is an independent set of size $k+1$ in $G'$ if $I$ is an independent set of size $k$ in $G$. Vice versa, given an independent set $I'$ of size $k+1$ in $G'$, at most one of its vertices belongs to
	$X$ (as $X$ is a clique) and thus  $I'\setminus X$ is an independent set of size at least $k$ in $G$.
	
Finally, we observe that the graph $G'$ can be constructed in time $\mathcal{O}(m^2)$ and has at most $n+m^2$ vertices. 
\end{proof}

\noindent
We now strengthen Theorem~\ref{t-al82}, but can only do this for the case where $|{\cal F}|=1$. For example, if ${\cal F}=\{C_3,K_{1,4}\}$, every ${\cal F}$-free graph has maximum degree~$3$. Then every ${\cal F}$-free graph with bounded diameter has constant size. Hence, {\sc Independent Set} is \NP-complete for ${\cal F}$-free graphs by Theorem~\ref{t-al82} but constant-time solvable for ${\cal F}$-free graphs of bounded diameter.

\begin{theorem}\label{t-stronger}
	Let $H$ be a graph. If $H\notin {\cal S}$, then {\sc Independent Set} is \NP-complete for $H$-free graphs of diameter at most~$2$. 
	If $H\in {\cal S}$, then {\sc Independent Set} for $H$-free graphs is polynomially equivalent to {\sc Independent Set} for $H$-free graphs of diameter at most~$2$.
\end{theorem}

\begin{proof}
	Let  $H$ be a graph. First assume that $H\notin {\cal S}$. If $H$ contains an induced cycle~$C_r$ for some $r\geq 3$, then the class of $H$-free graphs contains the class of $C_r$-free graphs, and we use Lemma~\ref{l-ic3d2}.
	Otherwise $H$ is a forest with either an induced $K_{1,4}$ or a connected component that has at least two vertices of degree~$3$. In the first case, 
	we use Lemma~\ref{l-k14}. In the second case, we reduce from {\sc Independent Set} for $H$-free graphs, which is \NP-complete by Theorem~\ref{t-al82}. Let $(G,k)$ be an instance of {\sc Independent Set}, where $G$ is an $H$-free graph. 
We add a dominating vertex to $G$ to obtain a graph $G'$. As $H$ has no dominating vertex, $G'$ is $H$-free.
We also note that $(G,k)$ and $(G',k)$ are equivalent instances.
	
Now assume that $H\in {\cal S}$. Any polynomial-time algorithm for {\sc Independent Set} on $H$-free graphs can be used on $H$-free graphs of diameter~$2$. As {\sc Independent Set} is polynomial-time solvable for $K_{1,3}$-free 
	graphs~\cite{Sb80}, we may assume that $H\notin \{K_{1,3},P_1,P_2,P_3\}$. Any polynomial-time algorithm for {\sc Independent Set} on $H$-free graphs of diameter~$2$ can be used on $H$-free graphs of arbitrary diameter as follows. To the $H$-free input graph $G$ we add a dominating vertex. As $H\in {\cal S}\setminus \{K_{1,3},P_1,P_2,P_3\}$
and $K_{1,3}$, $P_1$, $P_2$, $P_3$ are the only graphs in ${\cal S}$ with a dominating vertex,	
this yields an $H$-free graph $G'$ of diameter~$2$. We then observe that $G$ has an independent set of size at least~$k$ if and only if $G'$ has an independent set of size at least~$k$.
\end{proof}

\section{Chair-Free Graphs of Bounded Diameter}\label{s-h}

It follows from Ramsey's Theorem that every $k$-colourable $K_{1,r}$-free graph of diameter at most~$d$ has order bounded by a function in $d,k,r$; see~\cite{MPS19} for a proof of this observation. As a consequence, every problem that has the property that all its yes-instances are $k$-colourable for some constant~$k$ is constant-time solvable on $K_{1,r}$-free graphs of diameter at most $d$.
We aim to extend the above observation to $H$-free graphs of bounded diameter when $H$ is obtained from a star after at least one edge subdivision.
In~\cite{MPS19}, a number of results are given for {\sc $3$-Colouring} for such graph classes. 
We consider a variety of problems that require all the yes-instances to be $3$-colourable. We focus on the first interesting case which is where $H$ is the chair $S_{1,1,2}$ (recall that the chair is obtained from the claw after subdividing one edge). 

We need the following characterization of bipartite chair-free graphs, due to Alekseev.
A \emph{complex} is a bipartite graph that can be obtained by removing the edges of a possibly empty matching from a complete bipartite graph.

\begin{theorem}[\cite{Al04}]\label{t-al04}
If $G$ is a connected bipartite chair-free graph, then $G$ is a cycle or a path or a complex.
\end{theorem}

It is well-known (cf.\,\cite{HoTa73}) that finding the components of a graph by breadth-first search takes $O(n+m)$ time. Let $p$ be the number of components of a graph $G$, $n$ be its order, and $m$ be its size. Then $G$ is a forest if and only if $p=n-m$, which implies the following result.

\begin{observation}\label{c-forest}
If $G$ is a graph, then we can decide if $G$ is a forest in $O(n+m)$ time. 
\end{observation}

If $T$ is a tree of order $n$, then its diameter is at most $2$ if and only if its maximum degree equals $n-1$. Therefore, we can decide whether a given graph is a forest each component of which is of diameter at most $2$ in $O(n+m)$ time. When working with vertex labellings, our findings imply the following observation.

\begin{observation}\label{c-3-colourings}
If $G$ is a graph and $\ell$ is a vertex labelling of $G$ with labels $1,2,$ and $3$, then we can decide whether $\ell$ is a $3$-colouring, star $3$-colouring, or acyclic $3$-colouring of $G$ in $O(n+m)$ time. 
\end{observation}

It is also well-known that we can use breadth-first search for deciding whether a given graph $G$ is bipartite and, if so, we are in a position to determine its parts in the same time. By this fact, we obtain the following result.

\begin{observation}\label{c-k}
If $G$ is a graph, $k$ is an integer, and $\ell$ is a vertex labelling of $G$ with labels $1,2,$ and $3$, then we can decide whether one colour class of $\ell$ is an independent feedback vertex set or an independent odd cycle transversal (of size at most $k$) in $O(n+m)$ time.
\end{observation}

To prove our results we need some more terminology.
A {\it list assignment} of a graph~$G$ is a function $L$ that gives each vertex $u\in V(G)$ a (finite) {\it list of admissible colours} $L(u)\subseteq \{1,2,\ldots\}$.
A colouring $c$  {\it respects} ${L}$ if  $c(u)\in L(u)$ for every $u\in V(G).$ If $|L(u)|\leq 2$ for each $u\in V(G)$, then $L$ is 
a {\it $2$-list assignment}. The {\sc $2$-List Colouring} problem is to decide if a graph $G$ with a $2$-list assignment $L$ has a colouring that respects $L$. 

We use the following well-known result, which is due to Edwards (see the proof of Lemma~2.4 in~\cite{Ed86}).

\begin{lemma}[\cite{Ed86}]\label{l-2sat}
The {\sc $2$-List Colouring} problem is solvable in $O(n+m)$ time on graphs with $n$ vertices and $m$ edges.
\end{lemma}

We now introduce some additional terminology. 
Let $G$ be a graph of diameter $d$ for some $d\geq 1$ and $S$ be some set of vertices of $G$. A vertex $u\notin S$ is a {\it private} neighbour of a vertex $v\in S$ {\it with respect to} $S$ if $u$ is adjacent to $v$ but non-adjacent to any other vertex of $S$. We let $P(v)$ be the set of private neighbours of $v$ with respect to $S$.
Let $N_0=S$ and, for $i\in \{1,\ldots d\}$, let $N_i$ be the set of vertices that do not belong to $N_0\cup \ldots \cup N_{i-1}$ but that do have a neighbour in $N_{i-1}$; in particular $N_1=N(S)$. As $G$ has diameter at most~$d$, we partition $V(G)$ into the sets $N_0$, $N_1$, \ldots, $N_{d}$ (where some sets might be empty). We say that we {\it partition $V(G)$ from $S$}.
Note that this partitioning takes $O(n+m)$ time by breadth-first search on the graph $G'$ which is obtained from $G$ by adding a new vertex $u$ and all edges from $u$ to every vertex of $S$.

In the next lemma we show that a sufficiently large connected chair-free graph contains a triangle $T$, which we can find in linear time. In later proofs we will partition from $V(T)$.

\begin{lemma}\label{l-3-col-1}
Let $d\geq 1$ and $G$ be a chair-free non-bipartite graph of diameter $d$ with $n$ vertices and $m$ edges. If $G$ has at least $2d+2$ vertices, then $G$ contains at least one triangle, which is computable in $O(n+m)$ time.
\end{lemma}
\begin{proof}
For a contradiction, assume that $G$ is triangle-free. As $G$ is not bipartite, there is an odd cycle in $G$. Let $x_1x_2\ldots x_px_1$ be a shortest one. As $G$ is triangle-free and of diameter~$d$, we find that $5\leq p\leq 2d+1$. 
Moreover, as $G$ is of order at least $2d+2$, there is some vertex outside this cycle that has a neighbour on this cycle. Without loss of generality let us assume $y$ with $y\notin \{x_1,x_2,\ldots,x_p\}$ is adjacent to $x_1$.
 As $G$ is triangle-free, $y$ does not have two consecutive neighbours on $x_1x_2\ldots x_px_1$. 
 As $G$ is chair-free and $y$ is neither adjacent to $x_2$ nor to $x_p$, we find that $y$ must be adjacent to $x_3$. We repeat this argument and obtain that  
 $y$ is adjacent to $x_{2q+1}$ for every $ 0 \leq q \leq  \lfloor \frac{p}{2} \rfloor$. In particular $y$ is adjacent to the two consecutive vertices $x_1$ and $x_p$, a contradiction.
 We conclude that our assumption is false and that $G$ contains a triangle.

We continue and show that we can compute a triangle of $G$ in $O(n+m)$ time. Let $u$ be a vertex of $G$. We partition $V(G)$ from $\{u\}$ and note that breadth-first search computes a breadth-first tree $F$, that is, $F$ is a spanning tree of $G$ such that each vertex of $N_i$ has distance $i$ to $u$ in $F$ for any~$i$. As $G$ is not bipartite, there has to be an edge $e$ and an integer $i$ such that $e$ is incident to two vertices of $N_i$. We can compute such an edge that additionally minimizes $i$ in $O(n+m)$ time. By adding this edge to $F$, we find an odd cycle~$C$ in $G$. As $F$ is of diameter at most~$2d$, we find that $C$ has at most $2d+1$ vertices. Hence, we can determine in constant time an induced odd cycle, say $C'$, in $G[V(C)]$. 
We check in constant time whether $C'$ is a triangle. If not, then $C'$ is of order at least $5$. 
As $G$ is of order at least $2d+2$, there is a vertex outside $C'$ that has a neighbour on $C'$. We compute such a vertex, say $y$, in $O(n+m)$ time. As shown above, $y$ has two consecutive neighbours on $C'$. As $C'$ has at most~$2d+1$ vertices, we can find such two vertices, and thus a triangle in $G$, in constant time.
\end{proof}

\noindent
In the next two lemmas we show two necessary conditions for a chair-free graph $G$ with a triangle $T$ to be $3$-colourable. To prove the second lemma, we partition $V(G)$ from $V(T)$.

\begin{lemma}\label{l-3-col-2}
Let $G$ be a $3$-colourable chair-free graph, and let $T$ be a triangle in $G$ with vertex set $\{x,y,z\}$. If $T$ contains at least two vertices that have a private neighbour, then $|P(x)\cup P(y)\cup P(z)|\leq 6$.
\end{lemma}

\begin{proof}
For a contradiction, assume that $|P(x)\cup P(y)\cup P(z)|> 6$. By the pigeonhole principle, one vertex of $T$, say $x$, has at least three private neighbours. As $G$ is $3$-colourable, and thus $K_4$-free, there are two non-adjacent vertices in $P(x)$, say $u$ and $v$.
As at least two vertices of $T$ have a private neighbour, we may assume without loss of generality that $y$ has a private neighbour $w$.
Then $\{u,v,w,x,y\}$ induces a chair unless $w$ is adjacent to at least one of $u$ and $v$. If $w$ is adjacent to exactly one of $u$, $v$, then $\{u, v, w, x, z\}$ induces a chair. If $w$ is adjacent to both $u$ and $v$, then $\{u,v,w,y,z\}$ induces a chair. In all three cases, we obtained a contradiction.
\end{proof}

\begin{lemma}\label{l-3-col-3}
Let $G$ be a $3$-colourable chair-free graph of diameter $d$ for some $d\geq 1$, and let $T$ be a triangle in $G$ with vertex set $\{x,y,z\}$. Then $G-N(T)$ has at most $9\cdot 2^d +2$ vertices.
\end{lemma}
\begin{proof}
For a contradiction, assume that $G$ has at least $9\cdot 2^d +3$ vertices. 
We partition $V(G)$ from $V(T)$.
As $G$ is $3$-colourable, $G$ is $K_4$-free. Hence, every vertex of $N_1$ has at least one non-neighbour on $T$.  
Let $i\geq 1$ and let $u$ be a vertex of $N_i$. As $G$ is $\chair$-free, the neighbours of $u$ in $N_{i+1}$ form a clique. 
Note that for $i=1$ this holds since $u$ must have at least one neighbour and at least one non-neighbour on~$T$.
As $G$ is $K_4$-free, this means that the neighbourhood of 
 $u$ in $N_{i+1}$ is a clique of size at most~$2$. It follows that
\[3+ 9\cdot 2^d\leq  |N_0|+|N_2|+|N_3|+\ldots+|N_d|\leq 3+ |N_2|\cdot \sum_{i=2}^d 2^{i-2}< 3+|N_2|\cdot 2^{d-1}.\]
Hence, $|N_2|>18$. We let $N_1^*$ be the set of all vertices of $N_1$ that have two neighbours on~$T$, and we let $N_2^*$ be the set of neighbours of $N_1^*$ in $N_2$.
Consider the set $N_{xy}$ of common neighbours of $x$ and $y$ in $N_1^*$. Then $N_{xy}$ is an independent set, as otherwise two adjacent vertices in~$N_{xy}$ form, together with $x$ and $y$, a $K_4$, contradicting the $K_4$-freeness of $G$.

Every vertex $u\in N_2^*$ with a neighbour  in $N_{xy}$ must be adjacent to every vertex in~$N_{xy}$, as~$G$ is chair-free and $N_{xy}$ is an independent set. 
Recall that every vertex of $N_1$, and thus every vertex of $N_{xy}$, has at most two neighbours in $N_2$, which belong to $N_2^*$ by definition. Hence, there are at most two vertices in $N_2^*$ that are adjacent to the vertices of $N_{xy}$. By applying the same reasoning for every other pair of vertices of $T$, we find that $N_2^*$ has size at most~$6$.  
Thus, $|N_2\setminus N_2^\star|>12$. As every vertex of $N_1$ has at most two neighbours in $N_2$, it follows that $|P(x)\cup P(y)\cup P(z)|>6$. 
By Lemma~\ref{l-3-col-2}, we obtain that there is exactly one vertex, say $x$, of $T$ which has private neighbours.
As $G$ is $3$-colourable, $G[P(x)]$ is bipartite. 

As $G[P(x)]$ is bipartite, we can partition $P(x)$ into two independent sets $A$ and $B$ (one of these two sets might be empty).
As $G$ is chair-free and as $A$ is independent, the vertices of $A$ share the same set of neighbours in $N_2$. Similarly, the vertices of $B$ share the same set of neighbours in $N_2$. Recall that the neighbourhood of every vertex of $A\cup B$ in $N_2$ is a clique of size at most~$2$. We conclude that the total number of vertices in $N_2$ with a neighbour in $P(x)$ is at most~$4$, a contradiction as $|N_2\setminus N^*_2|>12$.
We conclude that $G-N_1$ has at most $9\cdot 2^d +2$ vertices.
\end{proof}

\noindent
In~\cite{MPS19}, it was shown that {\sc $3$-Colouring} is polynomial-time solvable for chair-free graphs of bounded diameter. By giving more precise arguments, we show in the first statement of Theorem~\ref{t-3-col} that we can even obtain a linear-time algorithm for {\sc $3$-Colouring} on chair-free graphs of bounded diameter. Recall that we wish to prove a similar result for a number of vertex partitioning problems that require yes-instances to be $3$-colourable. In order to do this (in Theorem~\ref{t-chair}), we first need to prove a new result on the structure of $3$-colourable chair-free (non-bipartite) graphs of bounded diameter and on how to find this structure.
These graphs might have an exponential number of $3$-colourings (for example, take a triangle~$T$ and add a number of pendant vertices to one of the vertices of $T$). However, if there are too many, then we find a certain induced subgraph that has a bounded number of $3$-colourings. This is shown in the second statement of Theorem~\ref{t-3-col}.
The proof of Theorem~\ref{t-3-col} uses Lemmas~\ref{l-3-col-1}--\ref{l-3-col-3}.

\begin{theorem}\label{t-3-col}
For an integer $d\geq 1$ and a chair-free non-bipartite graph $G$ of diameter~$d$ with $n$ vertices and $m$ edges, it is possible to find the following in $O(n+m)$ time:
\begin{enumerate}
\item 
whether or not $G$ is $3$-colourable, and 
\item 
should $G$ be $3$-colourable, either all $3$-colourings of $G$, or a triangle $T$ with exactly one vertex~$x$ that has private neighbours and all $3$-colourings of $G-P(x)$ that can be extended to $3$-colourings of $G$; in both cases, the number of $3$-colourings found is at most $3^{9\cdot 2^d+8}$.  
\end{enumerate}
\end{theorem}

\begin{proof}
We first check in constant time whether $G$ is of order at most $2d+1$. If so, then we can determine in constant time all $3$-colourings of $G$ and these are at most $3^{2d+1}$. Note that $3^{2d+1}<3^{9\cdot 2^d+8}$.
In what follows, we assume that $G$ is of order at least $2d+2$. 

By Lemma~\ref{l-3-col-1}, $G$ contains a triangle and we can compute such a triangle, say $T$, in $O(n+m)$ time. Let $x,y,z$ be the vertices of $T$. We partition $V(G)$ from $\{x,y,z\}$ in $O(n+m)$ time.
We additionally determine all private neighbours of the vertices of $T$ and all vertices of $N_1$ that are adjacent to all vertices of $T$ in linear time.
If there is a vertex of the latter type, then $G$ is not $3$-colourable. Thus, we focus on the case where each vertex of $N_1$ is adjacent to at most two vertices of $T$. 
For simplicity, let $S=P(x)\cup P(y)\cup P(z)$.

We compute the set $N_1^*$ of all vertices of $N_1$ that have two neighbours on $T$. Clearly, $S=N_1\setminus N_1^*$ and the computation of $N_1^*$ and $S$ takes $O(n+m)$ time.
We proceed by considering $G-N_1$. We determine $|V(G-N_1)|$ in $O(n+m)$ time and check if $|V(G-N_1)|\geq 
9\cdot 2^d+3$. If so, then by Lemma~\ref{l-3-col-3} we find that $G$ is not $3$-colourable. Assume that $G-N_1$ has at most $9\cdot 2^d+2$ vertices.

We consider every vertex labelling of $G-N_1$ with labels $1,2,3$. In what follows we determine in $O(n+m)$ time which ones lead to a $3$-colouring of $G$. 
First of all, we discard those labellings which are not a $3$-colouring of $G-N_1$. 
Given a $3$-colouring of $G-N_1$, each vertex of $N_1^*$ receives the remaining available label that is not used for its neighbours on $T$. Note that this assignment takes linear time. We discard such a labelling if it does not lead to a $3$-colouring of $G-S$.
Next we assign lists to the vertices of $G$ as follows: we set $L(u)=\{i\}$, where $i$ is the label of $u$, if $u\notin S$ and we set $L(u)=\{1,2,3\}\setminus\{i\}$, where $i$ is the label of the unique neighbour of $u$ on $T$, if $u\in S$. Thus, checking whether a given $3$-colouring of $G-N_1$ leads to a $3$-colouring of $G$ takes $O(n+m)$ time as $(G,L)$ is an instance of {\sc $2$-List Colouring} (cf.\,Lemma~\ref{l-2sat}). 
We discard those $3$-colourings of $G-N_1$ which do not lead to a $3$-colouring of $G$.
As there are at most $3^{9\cdot 2^d+2}$ vertex labellings of $G-N_1$, we repeat this procedure constantly many times. Note that we get all $3$-colourings of $G-S$ as a by-product.

If all $3$-colourings of $G-N_1$ do not lead to a $3$-colouring of~$G$, then $G$ is not $3$-colourable. 
Consequently, we can decide whether $G$ is $3$-colourable in $O(n+m)$ time.
As Theorem~\ref{t-3-col}.1 is now proven, we proceed with our proof for Theorem~\ref{t-3-col}.2 by assuming that $G$ is $3$-colourable.

As there are at most $3^{9\cdot 2^d+2}$ vertex labellings of $G-N_1$, there are at most $3^{9\cdot 2^d+2}$ $3$-colourings of $G-S$.
If $S=\emptyset$, then $G-S$ equals $G$ and Theorem~\ref{t-3-col}.2 follows.
If all vertices of $S$ are adjacent to a single vertex, say $x$, of $T$, we conclude that $S=P(x)$ and Theorem~\ref{t-3-col}.2 follows.
Hence, we may assume that there are at least two vertices of $T$ which have a private neighbour.
As $G$ is $3$-colourable, we find that $|S|\leq 6$ by Lemma~\ref{l-3-col-2}.
Thus, as we have at most $3^{9\cdot 2^d+2}$ $3$-colourings of $G-S$, we find at most 
$3^{9\cdot 2^d+8}$ $3$-colourings of~$G$ and their computation takes $O(n+m)$ time, which completes our proof.
\end{proof}

\noindent
The double-exponential bound in Theorem~\ref{t-3-col} cannot be improved to single-exponential. In order to see this
we define a series $G_1,G_2,\ldots$ of connected chair-free graphs as follows (see Figure~\ref{f-expon} for a drawing of the graph $G_4$):

\begin{itemize}
\item $G_1$ is a triangle with vertex set $\{v_1^1,v_2^1,v_3^1\}$.
\item $G_{d+1}$ is obtained from $G_d$ by adding $3\cdot 2^d$ new vertices $\{v_1^{d+1},v_2^{d+1},\ldots,v_{3\cdot 2^d}^{d+1}\}$ and adding the edges 
$v_i^dv_{2i-1}^{d+1}, v_i^dv_{2i}^{d+1},v_{2i-1}^{d+1}v_{2i}^{d+1}$ for each $i=1,2,\ldots, 3\cdot 2^{d-1}$.
\end{itemize}

\noindent
We note that $G_d$ is of diameter $2d-1$ for each $d$.
Let $a(d)$ be the number of $3$-colourings of $G_d$.
We claim $a(d)=6\cdot 2^{3\cdot (2^{d-1}-1)}$ and prove this by induction. Obviously, $a(1)=6$.
 For each $3$-colouring of $G_d$ and each $i=1,2,\ldots, 3\cdot 2^{d-1}$, we have two possibilities to colour the vertices $v_{2i-1}^{d+1}$ and $v_{2i}^{d+1}$. Thus, using induction hypothesis, we find that
$a(d+1)=a(d)\cdot 2^{3\cdot 2^{d-1}}=\left(6\cdot 2^{3\cdot (2^{d-1}-1)}\right)\cdot 2^{3\cdot 2^{d-1}}=6\cdot 2^{3\cdot (2^{d}-1)}$.

Furthermore, the only triangles with exactly
one vertex that has private neighbours are those on vertex set $\{v_i^{d-1},v_{2i-1}^d,v_{2i}^d\}$ for $i=1,2,\ldots,3\cdot 2^{d-1}$. For such a triangle,  the vertex with private neighbours is $v_i^{d-1}$. As every two distinct $3$-colourings of $G_d$ differ on at least one vertex of 
$V(G_d-P(v_i^{d-1}))$, there are $a(d)$ $3$-colourings of $G_d-P(v_i^{d-1})$ that can be extended to $3$-colourings of $G_d$.

\begin{figure}[h]
\centering
\begin{tikzpicture}[scale=1.4]

\foreach \j in {1,2,3}
		\draw ({360/3*(\j+0.5)-90}:0.5) node[bnode]{} -- ({360/3*(\j+1.5)-90}:0.5) node[bnode]{};
\foreach \i[evaluate=\i as \square using int(3*2^\i)] in {1,2,3}{
	\foreach \j in {1,2,...,\square}
		\draw ({360/\square*2*(\j+0.5)-90}:{(\i-1)/1.5+0.5}) node[bnode]{}--
				({360/\square*((2*\j-1)+1.5)-90}:{\i/1.5+0.5}) node[bnode]{} -- 
				({360/\square*((2*\j)+1.5)-90}:{\i/1.5+0.5}) node[bnode]{} --
				({360/\square*2*(\j+0.5)-90}:{(\i-1)/1.5+0.5}) node[bnode]{};
}
\draw ({360/3*(1+0.5)-65}:0.5) node{$v_{1}^1$};
\draw ({360/3*(2+0.5)-65}:0.5) node{$v_{2}^1$};
\draw ({360/3*(3+0.5)-65}:0.5) node{$v_{3}^1$};
\draw ({360/6*(2.5)-101}:{1/1.5+0.5}) node{$v_{1}^2$};
\draw ({360/6*(3.5)-79}:{1/1.5+0.5}) node{$v_{2}^2$};
\draw ({360/6*(4.5)-101}:{1/1.5+0.5}) node{$v_{3}^2$};
\draw ({360/6*(5.5)-79}:{1/1.5+0.5}) node{$v_{4}^2$};
\draw ({360/6*(6.5)-101}:{1/1.5+0.5}) node{$v_{5}^2$};
\draw ({360/6*(7.5)-79}:{1/1.5+0.5}) node{$v_{6}^2$};
\draw ({360/12*(1+3.5)-97}:{2/1.5+0.5}) node{$v_{1}^3$};
\draw ({360/12*(2+3.5)-83}:{2/1.5+0.5}) node{$v_{2}^3$};
\draw ({360/12*(3+3.5)-97}:{2/1.5+0.5}) node{$v_{3}^3$};
\draw ({360/12*(4+3.5)-83}:{2/1.5+0.5}) node{$v_{4}^3$};
\draw ({360/12*(5+3.5)-97}:{2/1.5+0.5}) node{$v_{5}^3$};
\draw ({360/12*(6+3.5)-83}:{2/1.5+0.5}) node{$v_{6}^3$};
\draw ({360/12*(7+3.5)-97}:{2/1.5+0.5}) node{$v_{7}^3$};
\draw ({360/12*(8+3.5)-83}:{2/1.5+0.5}) node{$v_{8}^3$};
\draw ({360/12*(9+3.5)-97}:{2/1.5+0.5}) node{$v_{9}^3$};
\draw ({360/12*(10+3.5)-83}:{2/1.5+0.5}) node{$v_{10}^3$};
\draw ({360/12*(11+3.5)-97}:{2/1.5+0.5}) node{$v_{11}^3$};
\draw ({360/12*(12+3.5)-83}:{2/1.5+0.5}) node{$v_{12}^3$};
\foreach \j in {1,2,...,24}{
	\draw ({360/24*(\j+7.5)-90}:{3/1.5+0.72}) node{$v_{\j}^4$};
}
\end{tikzpicture}
	\caption{The graph $G^4$.}
	\label{f-expon}
\end{figure}

\medskip
\noindent
We now use Theorem~\ref{t-3-col} to  prove that apart from {\sc $3$-Colouring} (Theorem~\ref{t-3-col}.1), the problems
{\sc Acyclic $3$-Colouring}, {\sc Star $3$-Colouring}, {\sc Independent Odd Cycle Transversal}, {\sc Independent Feedback Vertex Set}, and  {\sc Near-Bipartiteness} can be solved in linear time for chair-free graphs of bounded diameter.

We first deal, in the next lemma, with the case that the chair-free input graph of bounded diameter is bipartite.

\begin{lemma}\label{l-bibbib}
 For every integer $d\geq 1$, {\sc Acyclic $3$-Colouring}, {\sc Star $3$-Colouring}, {\sc Independent Odd Cycle Transversal}, {\sc Independent Feedback Vertex Set}, and  {\sc Near-Bipartiteness} can be solved in $O(n+m)$ time for chair-free bipartite graphs of diameter at most $d$ with $n$ vertices and $m$ edges. 
 \end{lemma}

\begin{proof}
Note that a bipartite graph is near-bipartite and has an independent odd cycle transversal of size at most $k$ for every integer $k$. 
Hence, it remains to consider the problems {\sc Acyclic $3$-Colouring}, {\sc Star $3$-Colouring} and {\sc Independent Feedback Vertex Set}.

Let $G$ be a chair-free bipartite graph of diameter at most~$d$ with $n$ vertices and $m$ edges. We may assume without loss of generality that $G$ is connected.
We first determine the two partition classes  $S_1$ and $S_2$ of $G$ in $O(n+m)$ time. We may assume without loss of generality that $|S_1|\geq |S_2|$. 
We check in constant time whether $|S_1|+|S_2|\leq \max\{8,2d\}$ and if so, then we can solve each of our problems in constant time. Otherwise, we find that $|S_1|\geq 5$. 
As bipartite graphs of maximum degree at most~$2$ and diameter at most~$d$ are paths or cycles of at most $2d$ vertices, we find that $G$ has a vertex of degree at least $3$, and so $G$ is a complex by Theorem~\ref{t-al04}.

We first claim that in the case where $G$ is a complex with $|S_1|\geq 5$, $G$ is star $3$-colourable if $|S_2|\leq 2$ and acyclic $3$-colourable only if $|S_2|\leq 2$.
This claim suffices to prove the lemma for {\sc Acyclic $3$-Colouring} and {\sc Star $3$-Colouring}
as we can decide whether $|S_2|\leq 2$ or not in constant time and moreover, every star $3$-colouring of a graph is acyclic.

We prove the above claim as follows:
If $|S_2|\leq 2$, then, for any $s\in S_2$, $G-s$ is a forest each component of which is of diameter at most $2$, and thus $G$ is star $3$-colourable with colour classes $S_1,S_2\setminus \{s\}$, and $\{s\}$. 
If $|S_2|\geq 3$, then let $c$ be an arbitrary $3$-colouring of $G$. 
By the pigeonhole principle there exists a colour class $X$ of $c$ that contains at least two vertices of $S_1$, and so $X\cap S_2=\emptyset$.		
As $|S_2|\geq 3$, there are two vertices $s_2,s_2'\in S_2$ that are coloured alike. As $|S_1|\geq 5$, and as $s_2$ and $s_2'$ are of degree at least $|S_1|-1$, we find that $s_2$ and $s_2'$ have at least three common neighbours in $S_1$ two of which, say $s_1$ and $s_1'$, are coloured alike. Hence, $s_1s_2s_1's_2's_1$ is a bichromatic $4$-cycle. We conclude that every $3$-colouring of $G$ is not acyclic, which completes the proof of our claim.

It remains to consider {\sc Independent Feedback Vertex Set}. Let $k$ be an arbitrary integer.
We claim that if $G$ is a complex with $|S_1|\geq 5$, then $G$ has an independent feedback vertex set of size at most $k$ if and only if $k\geq |S_2|-1$. Note that the latter can be decided in linear time.

We prove the above claim as follows: If $|S_2|\leq 2$, then $G-s$ is a forest for any $s\in S_2$ and $G$ has an independent feedback vertex set of size at most $k$.
Hence, we may assume $|S_2|\geq 3$.
Let $I$ be a minimum independent feedback vertex set in $G$. Such a set exists as $G$ is bipartite. As $S_2\setminus \{s\}$ is independent and as $G[S_1\cup \{s\}]$ is a forest for each vertex $s\in S_2$, we find $|I|\leq |S_2|-1$. For the sake of a contradiction, let us assume $|I|\leq |S_2|-2$. Hence, any two vertices of $S_2\setminus I$ have at least $|S_1|-2$ common neighbours in $S_1$, and so $|I\cap S_1|\geq |S_1|-3\geq 2$. Moreover, $I=I\cap S_1$ as every vertex of $S_2$ has a neighbour in $I\cap S_1$ and as $I$ is independent. As $I$ is an independent feedback vertex set with $|I|\leq |S_1|-2$, any two vertices of $S_1\setminus I$ do not have two common neighbours in $S_2$ and so $|S_2|\leq 3$. Hence, 
$5\leq |S_1|\leq |I|+3\leq |S_2|+1\leq 4$, a contradiction. 
As $|I|=|S_2|-1$, the proof of the lemma is complete.
\end{proof}

\noindent
We are now ready to prove the main result of this section, which is the same statement as Lemma~\ref{l-bibbib} except that we now also include the non-bipartite graphs.

\begin{theorem}\label{t-chair} 
For every integer $d\geq 1$,  {\sc Acyclic $3$-Colouring}, {\sc Star $3$-Colouring}, {\sc Independent Odd Cycle Transversal}, {\sc Independent Feedback Vertex Set}, and  {\sc Near-Bipartiteness} can be solved in $O(n+m)$ time for chair-free graphs of diameter at most $d$ with $n$ vertices and $m$ edges. 
\end{theorem}

\begin{proof}
Let $G$ be a chair-free graph of diameter at most~$d$ with $n$ vertices and $m$ edges.
Recall that $G$ is acyclic $3$-colourable or star $3$-colourable only if $G$ is $3$-colourable. Moreover, if $I$ is an independent set of $G$ for which $G-I$ is a bipartite graph, then $G$ is $3$-colourable. Hence, our problems require all the yes-instances to be $3$-colourable.
If $d=1$, then $G$ is $3$-colourable if and only if $G$ has at most $3$ vertices, and so each of our problems can be solved in constant time. We proceed by assuming $d\geq 2$. We first check in $O(n+m)$ time whether $G$ is bipartite.
If so, then we apply Lemma~\ref{l-bibbib}. In the remainder of our proof we assume that $G$ is not bipartite.

\medskip
\noindent
{\it Outline.} As our problems require all the yes-instances to be $3$-colourable, we check first whether $G$ is $3$-colourable. If so, then we compute an induced subgraph $H$ of $G$ and determine the set $\mathcal{C}$ of all its $3$-colourings that can be extended to $3$-colourings of $G$. As we compute $H$ by applying Theorem~\ref{t-3-col}, we find that $|\mathcal{C}|\leq 3^{9\cdot 2^d+8}$. We then distinguish some subcases. In some of them we further branch by extending our $3$-colourings. However, in some of them we find that $H$ equals $G$, and so Observations~\ref{c-3-colourings} and~\ref{c-k} imply that our six problems are solvable in $O(n+m)$ time as $\mathcal{C}$ is of constant size. As an implicit step, we apply this finding whenever $H$ is the whole graph $G$.

\medskip
\noindent
{\it Full Proof.} We first apply Theorem~\ref{t-3-col}. We continue by assuming that $G$ is $3$-colourable. In fact, the only remaining case is that where the lemma provides a triangle $T$ on vertex set $\{x,y,z\}$, a vertex $x$ of $T$ that has private neighbours, and the set of all $3$-colourings of $G-P(x)$ that can be extended to $3$-colourings of $G$. Note that we have at most $3^{9\cdot 2^d+8}$ such $3$-colourings. 
We partition $V(G)$ from $V(T)$.

We find that $G[P(x)]$ is bipartite, as $G$ is $3$-colourable, but not necessarily connected. We extend each $3$-colouring of $G-P(x)$ that can be extended to a $3$-colouring of $G$ to some vertices of $P(x)$.
Let $c$ be an arbitrary $3$-colouring of $G-P(x)$ that can be extended to a $3$-colouring of $G$. 
For $i\in\{0,1,2\}$, we compute in $O(n+m)$ time the set $S_i$ of all vertices of $P(x)$ which have $i$ available colours with respect to $c$, that is, $S_i$ is the set of all vertices of $P(x)$ which have neighbours in $3-i$ colours. 
As $c$ can be extended to a $3$-colouring of $G$, we find that $S_0$ is empty.
It takes $O(n+m)$ time to determine the available colour of each vertex in $S_1$. Furthermore, we can extend $c$ by breadth-first search in the same time to the vertices of those components of $G[P(x)]$ that contain at least one vertex of $S_1$. 

Let $S_c$ be the set of vertices that induce those components of $G[P(x)]$ that do not contain a vertex of $S_1$. Note that $S_c$ can be computed in $O(n+m)$ time
and that all neighbours of all vertices of $S_c$ in $V(G)\setminus S_c$ are coloured alike.
Moreover, every vertex of $S_c$ has its neighbours in $[N(y)\cap N(z)]\cup N_2\cup S_c\cup \{x\}$ by definition.
As $c$ can be extended to a $3$-colouring of $G$, we find that our approach leads to a $3$-colouring, say $c'$, of $G-S_c$. 
As there are at most $3^{9\cdot 2^6+2}$ $3$-colourings of $G-P(x)$, we find at most $3^{9\cdot 2^6+2}$ such triples $(c,c',S_c)$.
Furthermore, for each $3$-colouring $c_s$ of $G-P(x)$, there exists a triple $(c_s,{c_s}',S_{c_s})$ if $c_s$ can be extended to a $3$-colouring of $G$.
We proceed by considering the case where $S_c\neq\emptyset$ as otherwise $G=G-S_c$.
We continue by distinguishing on the problems we are considering. Recall that $G$ is $3$-colourable.

\medskip

\noindent 
\textbf{Case 1.} {\sc Acyclic $3$-Colouring} and {\sc Star $3$-Colouring}.

\noindent
We check whether for some triple $(c,c',S_c)$, the $3$-colouring $c'$ of $G-S_c$ that can be extended to an acyclic $3$-colouring or star $3$-colouring of $G$. By this approach, we clearly solve {\sc Acyclic $3$-Colouring} and {\sc Star $3$-Colouring}.

Let $(c,c',S_c)$ be an arbitrary triple as defined above. Recall that a star $3$-colouring of a graph is acyclic.
In time $O(n+m)$, we can determine the components of $G[S_c]$ and check whether $G[S_c]$ is a forest. If not, then $G[S_c\cup \{x\}]$, and thus $G$, is not acyclic $3$-colourable.
We continue and assume that $G[S_c]$ is a forest.
We check in $O(n+m)$ time if a vertex of $S_c$ has a neighbour in $N(y)\cap N(z)$. If so, say $s\in S_c$ is adjacent to $v\in N(y)\cap N(z)$, then $c'$ cannot be extended to an acyclic $3$-colouring of $c$ as either $s$ and $x$ are coloured alike or one of $\{svyxs,svzxs\}$ is a bichromatic $4$-cycle.

We proceed by assuming that $S_c$ has its neighbours in $N_2\cup S_c\cup \{x\}$.
As $G$ is chair-free, every two non-adjacent vertices of $S_c$ share the same neighbours in $N_2$ and, if there exists such a neighbour, then these two vertices have to be coloured differently to avoid a bichromatic $4$-cycle.
Therefore, in any acyclic extension of $c'$ to $G$, each of the two colour classes in $S_c$ either has size at most $1$ or has no neighbour in $N_2$.
We check in constant time if $S_c$ is of size at most $2$. If so, then there are at most $4$ possibilities to extend $c'$ to a $3$-colouring of $G$ and for each we apply Observation~\ref{c-3-colourings}.
Hence, we may assume $|S_c|\geq 3$. 
We check in $O(n+m)$ time if a vertex of $S_c$ has a neighbour in~$N_2$. 

Let us consider the subcase where $s\in S_c$ has a neighbour, say $v$, in $N_2$. Let $G_s$ be the component of $G[P(x)]$ that contains $s$.
Note that there are at most two possibilities to extend $c'$ to the vertices of $G_s$. 
We check in linear time if $S_c\setminus V(G_s)$ is of size at least $2$. If so, say $s_1,s_2\in S_c\setminus V(G_s)$, then $v$ is a neighbour of $s,s_1$, and $s_2$. Thus, $xs_1'vs_2'x$ is a bichromatic $4$-cycle for two vertices $s_1'$ and $s_2'$ of $\{s,s_1,s_2\}$. We conclude that $c'$ cannot be extended to an acyclic $3$-colouring of $G$. Hence, we may assume $|S_c\setminus V(G_s)|\leq 1$, and so there are at most four possibilities to extend $c'$ to a $3$-colouring of $G$ each of which can be obtained in $O(n+m)$ time. We apply Observation~\ref{c-3-colourings} for each.

We proceed by assuming that no vertex of $S_c$ has a neighbour in $N_2$. In other words, each vertex of $S_c$ has its neighbours in $S_c\cup\{x\}$.
As $x$ is a cut-vertex of $G$, any extension of $c'$ to a $3$-colouring of $G$ is acyclic if and only if $c'$ is acyclic. Hence, we apply Observation~\ref{c-3-colourings} on $G-S_c$ and $c'$ in order to solve {\sc Acyclic $3$-Colouring}. 

We now check in $O(n+m)$ time if each component of $G[S_c]$ is of diameter at most~$2$. If not, then $G[S_c\cup\{x\}]$, and thus $G$ is not star $3$-colourable. Let us proceed by assuming that each component of $G[S_c]$ is of diameter at most $2$. We find that every $3$-colouring of $G[S_c\cup \{x\}]$ is a star $3$-colouring. In other words, we can restrict ourselves to those $3$-colouring extensions of $c'$ to $G$ that assign one colour to all vertices of $S_c$ if $S_c$ is independent, and an arbitrary $3$-colouring extensions of $c'$ to $G$ if $S_c$ is not independent. Note that we can check in $O(n+m)$ time whether $S_c$ is independent. We find in both subcases at most two extensions of $c'$ to $G$ and apply Observation~\ref{c-3-colourings} for each in order to solve {\sc Star $3$-Colouring}. 

\medskip
\noindent 
\textbf{Case 2.} {\sc Independent Odd Cycle Transversal}.

\noindent
Let $k$ be an arbitrary integer.
We check whether some triple $(c,c',S_c)$ consists of a $3$-colouring $c'$ of $G-S_c$ that can be extended to a $3$-colouring of $G$ whose one colour class is an independent odd cycle transversal of size at most $k$. As all the yes-instances require $G$ to be $3$-colourable, this approach clearly solves {\sc Independent Odd Cycle Transversal}.

Let $(c,c',S_c)$ be an arbitrary triple as defined above. 
Moreover, let $X,Y,Z$ be the colour classes of $c'$ with $x\in X$, $y\in Y$, and $z\in Z$. Clearly, $X,Y,$ and $Z$ can be computed in linear time.
We decide in linear time which of $\{Y,Z\}$ is of smaller size, say $|Y|\leq |Z|$.

Recall that all vertices of $S_c$ have their neighbours in $S_c\cup X$. 
Note that $c'$ can be extended to a $3$-colouring of $G$ by $2$-colourings of $G[S_c]$ on the colours that $c'$ assigns to $y$ and $z$, and these are the only possibilities. 
We find that the smallest possible colour class of a $3$-colouring of $G$ that extends $c'$ consists of the vertices either in $X$ or in $Y\cup W$, where
$W$ is the smallest possible colour class of a $2$-colouring of $G[S_c]$. As we can compute the components of $G[S_c]$ and its parts in $O(n+m)$ time, we find $W$ in the same time. Hence, the smallest possible independent odd cycle transversal of $G$ that is a colour class of an extension of $c'$ to a $3$-colouring of $G$ is of size $\min\{|X|,|Y\cup W|\}$. We can compare the sizes of $X$ and $Y\cup W$ with $k$ in linear time. 

\medskip

\noindent 
\textbf{Case 3.} {\sc Independent Feedback Vertex Set} and  {\sc Near-Bipartiteness}.

\noindent
Let $k$ be an arbitrary integer.
We check whether some triple $(c,c',S_c)$ consists of a $3$-colouring $c'$ of $G-S_c$ that can be extended to a $3$-colouring of $G$ whose one colour class is an independent feedback vertex set (of size at most $k$). As all the yes-instances require $G$ to be $3$-colourable, this approach clearly solves {\sc Independent Feedback Vertex Set} and  {\sc Near-Bipartiteness}.

Let $(c,c',S_c)$ be an arbitrary triple as defined above. 
Moreover, let $X,Y,Z$ be the colour classes of $c'$ with $x\in X$, $y\in Y$, and $z\in Z$. Clearly, $X,Y,$ and $Z$ can be computed in linear time.
We check first whether $G-X$ is a forest in $O(n+m)$ time. If so, then we find that $X$ is an independent feedback vertex set of $G$ and we can determine its size in linear time. Hence, we proceed by assuming that $G-X$ contains a cycle or $|X|>k$. As we aim to find an extension of $c'$ to a $3$-colouring of $G$ whose one colour class is an independent feedback vertex set (of size at most~$k$), we find that such a set consists of the vertices of $Y$ or of $Z$, and the vertices of some set $A\subseteq S_c$.

Recall that all vertices of $S_c$ have their neighbours in $[N(y)\cap N(z)]\cup N_2\cup S_c\cup \{x\}$ and their neighbours in $[N(y)\cap N(z)]\cup N_2\cup \{x\}$ form an independent set. 
Note that $c'$ can be extended to a $3$-colouring of $G$ by $2$-colourings of $G[S_c]$ on the colours that $c'$ assigns to $y$ and $z$, and these are the only possibilities. 
If $G[S_c]$ is connected, which can be tested in $O(n+m)$ time, then there are at most two such possibilities, and so we apply Observation~\ref{c-k} for each.
We proceed by assuming that $G[S_c]$ is disconnected, and so $|S_c|\geq 2$.

We claim that all vertices of $S_c$ have the same neighbours in $N_2$. Let us assume that $v$ is an arbitrary vertex of $N_2$ that is adjacent to some vertex of $S_c$.
Let $S_v$ be the set of neighbours of $v$ in $S_c$.
By definition, we find that $S_v$ is non-empty.
As $G$ is chair-free, we obtain that every vertex of $S_v$ is adjacent to every vertex of $S_c\setminus S_v$ as otherwise $\{s_1,s_2,v,x,y\}$ would induce a chair for some possible vertices $s_1\in S_v$ and $s_2\in S_c\setminus S_v$. As $G[S_c]$ is disconnected, we find that $S_c\setminus S_v=\emptyset$, which completes the proof of our claim as $v$ is arbitrarily chosen.

We can check if there is a vertex in $N(y)\cap N(z)$ in $O(n+m)$ time. 
First assume there is such a vertex, say $w$. As $\{s_1,s_2,w,x,y\}$ does not induce a chair for each two vertices $s_1,s_2$ of an independent set $I$ of $G[S_c]$, we find that $w$ is adjacent to all but at most one vertex of $I$. As $G[S_c]$ is bipartite, it follows that $w$ has at least $|S_c|-2$ neighbours in~$S_c$.
For each $s\in N(w)\cap S_c$, we find $s\in A$ as $sxyws$ and $sxzws$ are $4$-cycles.
Note that $N(w)\cap S_c$ can be computed in $O(n+m)$ time.
As $|N(w)\cap S_c|\geq |S_c|-2$, we find at most eight possibilities to extend $c'$ to a $3$-colouring of $G$ by a $2$-colouring of $G[S_c]$ in which one colour class contains all the vertices of $N(w)\cap S_c$. We apply Observation~\ref{c-k} for each.
Hence, we may assume that $N(y)\cap N(z)=\emptyset$, and so every two vertices of $S_c$ share the same neighbours in $V(G)\setminus S_c$.

If no vertex of $N_2$ has a neighbour in $S_c$, then $x$ is a cut-vertex. In this case we find that $G$ has an independent feedback vertex set of size at most $k$ if and only if $G-S_c$ has an independent feedback vertex set (of size at most $k-|W|$, where
$W$ is the smallest possible colour class of a $2$-colouring of $G[S_c]$.
As $W$ can be computed in $O(n+m)$ time, we apply Observation~\ref{c-k} for $G-S_c$ and $c'$.

We proceed by considering the situation where $v\in N_2$ has a neighbour in $S_c$.
Recall that all vertices of $S_c$ are adjacent to $v$.
As $xs_1vs_2x$ is a $4$-cycle for any two vertices $s_1,s_2\in S_c$, we find that $A$ has size at least $|S_c|-1$. 
In other words, we aim for such a $2$-colouring of $G[S_c]$ whose one colour class is of size at most $1$.
If $S_c$ is not independent, we have at most two such possibilities, and each leads to a $3$-colouring of $G$. We apply Observation~\ref{c-k} for each.
Now suppose that $S_c$ is independent.
We find that any two vertices of $S_c$ have the same neighbours in $G$. Let us fix one vertex, say, $s$ of $S_c$. As there is at most one vertex of $S_c$ that is not in the independent feedback vertex set, we may assume that $s$ is that vertex.
We have four ways of colouring the vertices of $S_c$ such that all vertices of $S_c\setminus\{s\}$ receive the same colour. It remains to apply Observation~\ref{c-k} for each case.
\end{proof}

\section{Polyad-Free Graphs of Bounded Diameter}\label{s-limitations}

In this section we show that for bounded diameter tractability, we should not seek to extend some of the results from the previous section to omission of arbitrary polyads (subdivided stars).  We let $K_{1,r}^\ell$ denote the {\it $\ell$-subdivided $r$-star}, which is  the  graph obtained from a star~$K_{1,r}$ by subdividing {\it one} edge of $K_{1,r}$ exactly $\ell$ times.

The problems for which we show \NP-hardness when restricted to $H$-free graphs for some polyad $H$ are {\sc Independent Odd Cycle Transversal}, {\sc Acyclic $3$-Colouring} and {\sc Star $3$-Colouring}. We show this in Theorems~\ref{t-hard1}--\ref{t-hard3}, respectively.

In our proofs we reduce from two special variants of the {\sc Not-All-Equal $3$-Satisfiability} problem, which is well-known to be \NP-complete~\cite{Sc78}. The problem is defined as follows. Given a CNF formula $\phi$ that consists of a set $X= \{x_1,x_2,...,x_n\}$ of logical variables, and a set $C = \{C_1, C_2, . . . , C_m\}$ of three-literal clauses over $X$, does there exist a truth assignment for $X$ such that each $C_j$ contains at least one true literal and at least one false literal? If such a truth assignment exists, then $\phi$ is {\it not-all-equal satisfiable}.

Before explaining the two variants, we need to introduce some additional terminology.
Let $\phi$ be an instance of {\sc Not-All-Equal $3$-Satisfiability} in which every clause still has 
three literals, but 
each literal appears in at most two clauses.
A \emph{collection} $\mathcal{C}$ of $\phi$ is a set of pairs $(z,C)$ such that

\begin{itemize}
\item $z$ is a literal of the clause $C$, and
\item for every $z$ that appears in at least one clause of $\phi$, there exists exactly one pair $(z,C)\in {\cal C}$.
\end{itemize}

\noindent
A clause $C$ is \emph{uncovered} by ${\cal C}$ if $(z,C)\notin\mathcal{C}$ for each literal $z$ that appears in $C$.
The collection~$\mathcal{C}$ is \emph{covering} if no clause of $\phi$ is uncovered.

We can now discuss the aforementioned two variants of {\sc Not-All-Equal $3$-Satisfiability}, which we call variant~$A$ and $B$, respectively. In variant~A, every clauses still has 
three literals, but 
each literal appears in at most two clauses. Furthermore, variant A requires as additional input a covering collection.
In variant~B, every clause also still has three literals, but each clause has only positive literals and each literal  occurs  in  at  most  four different  clauses. For variant B we do not require the input to include a covering collection.
Variants~A and~B are  both \NP-complete. For variant A, we prove this in our next lemma, whereas the \NP-completeness of variant~B was shown by Darmann and D{\"{o}}cke~\cite{DD20}.
 
 \begin{lemma}\label{l-a}
 Variant~A of {\sc Not-All-Equal $3$-Satisfiability} is \NP-complete.
 \end{lemma}
 
 \begin{proof}
 We reduce from {\sc Not-All-Equal $3$-Satisfiability} in two steps.
 
 \medskip
 \noindent
{\bf Step 1.}
Let $\phi$ be a CNF formula with $m$ clauses and $n$ variables.
 Suppose that literal~$z$ appears in $k\geq 3$ clauses, say $z$ appears in clauses $C_1,C_2,\ldots,C_k$. From $\phi$ we first construct a CNF formula $\phi'$ as follows:
 
\begin{itemize}
\item Add three new variables $y_z$, $z_1$, and $z_2$.
\item Add three new clauses $(y_z, z_1,\bar z_2)$, $(y_z, \bar z_1, z_2)$ and $(z,\bar z_1,\bar z_2)$.
\item For $i=1,2$, replace $z$ by $z_i$ in $C_i$.
\end{itemize}

\noindent
We note that $\phi$ is not-all-equal satisfiable if and only if $\phi'$ is not-all-equal-satisfiable. The reason is that
$z_1$, $z_2$ and $z$ have the same value in any truth assignment of $\phi'$ for which each clause contains at least one true literal and at least one false literal.
We also note that $y_z, z_1,\bar z_1, z_2, \bar z_2$ each appear in at most two clauses of $\phi'$. Furthermore, $z$ appears in at most $k-1$ clauses of $\phi'$, namely, in $C_3,\ldots, C_k$ and $(z,\bar z_1, \bar z_2)$.
Hence, applying this reduction recursively (at most $O(n\cdot m)$ times) leads to an equivalent CNF formula of {\sc Not-All-Equal $3$-Satisfiability} of polynomial length, in which each literal appears in at most two clauses.

\medskip
\noindent
{\bf Step 2.}
Let $\phi$ be an instance of {\sc Not-All-Equal $3$-Satisfiability} with $m$ clauses and $n$ variables such that each literal appears in at most two clauses.

For every literal $z$ that appears in at least one clause, we fix one clause, say $C_z$. Let $\mathcal{C}$ be the set of all pairs $(z,C_z)$. Note that $\mathcal{C}$ is a collection of $\phi$ and that its computation takes $O(nm)$ time. Let $\gamma_\phi(\mathcal{C})$ be the number of clauses in $\phi$ which are uncovered.

If $\gamma_\phi(\mathcal{C})>0$, then let $C_1$ be an uncovered clause and $z$ be a literal that appears in~$C_1$.
By construction, there is a second clause $C_2$ in which $z$ appears as literal and $(z,C_2)\in \mathcal{C}$.
From $(\phi,\mathcal{C})$ we first construct a CNF formula $\phi'$ and a collection $\mathcal{C}'$ as follows:
 
\begin{itemize}
\item Add three new variables $y_z$, $z_1$, and $z_2$.
\item Add three new clauses $(y_z, z_1,\bar z_2)$, $(y_z, \bar z_1, z_2)$, $(z, \bar z_1, \bar z_2)$.
\item For $i=1,2$, replace $z$ by $z_i$ in $C_i$.
\item Remove the pair $(z,C_2)$.
\item Add six new pairs $(y_z,(y_z,\bar z_1, z_2))$, $(z,(z, \bar z_1, \bar z_2))$, $(z_1,C_1)$, $(z_2,C_2)$, $(\bar z_1, (z,\bar z_1, \bar z_2))$, and $(\bar z_2, (y_z, z_1, \bar z_2))$.
\end{itemize} 

\noindent
By repeating the arguments from the proof of Step~1, we find that $\phi$ is not-all-equal satisfiable if and only if $\phi'$ is not-all-equal-satisfiable. 
As before, it is readily seen that $\phi$ is an instance of {\sc Not-All-Equal $3$-Satisfiability} in which each literal appears in at most two clauses. Moreover, $\mathcal{C}'$ is a collection of $\phi'$.
This follows from the fact that ${\cal C}$ is a collection of $\phi$ and our construction: we replaced $(z,C_2)$ by a new pair that contains $z$ and added exactly one pair for the literals $y_z,z_1,z_2,\bar z_1, \bar z_2$.
We did this in such a way that $\mathcal{C}'$ covers all the new clauses as well as $C_1$ and $C_2$. Thus, $\gamma_{\phi'}(\mathcal{C}')<\gamma_\phi(\mathcal{C})$.
Hence, applying this reduction recursively (at most $O(m)$ times) leads to an equivalent CNF formula of {\sc Not-All-Equal $3$-Satisfiability} of polynomial length, in which each literal appears in at most two clauses and for which a covering collection is fixed.
 \end{proof}

\begin{theorem}\label{t-hard1}
{\sc Independent Odd Cycle Transversal} is \NP-complete for $K_{1,4}^3$-free graphs of diameter~$4$.
\end{theorem}

\begin{proof}
	We deploy an argument previously used in~\cite{MPS19}. Let us recall first the standard reduction to {\sc 3-Colouring} from {\sc Not-all-Equal $3$-Sat}. We use variant~A. Hence, we are given a CNF formula $\phi$ with clauses $C_1,\ldots, C_m$ and variables $x_1,\ldots,x_n$ such that each literal appears in at most two clauses. From $\phi$, we construct a graph $G$ as follows (see also Figure~\ref{fig:3-col}):
	
	\begin{itemize}
		\item Add two vertices $v_{x_i}$ and $v_{\bar x_i}$ for each variable $x_i$.
		\item Add an edge between $v_{x_i}$ and $v_{\bar x_i}$ for  $i=1,\ldots, n$.
		\item Add a new vertex $z$ and make $z$ adjacent to every $v_{x_i}$ and every $v_{\bar x_i}$ 
		\item For each clause $C_i$,
		add a triangle $T_i$ with three new vertices $c_{i_1}, c_{i_2}, c_{i_3}$. Fix an arbitrary  order of the literals of $C_i$, say $y_{i_1}, y_{i_2}, y_{i_3}$ where $y_{i_j}\in \{x_{i_j},\bar x_{i_j}\}$, and add the edge $v_{y_{i_j}}c_{i_j}$ for $j=1,2,3$.
	\end{itemize}

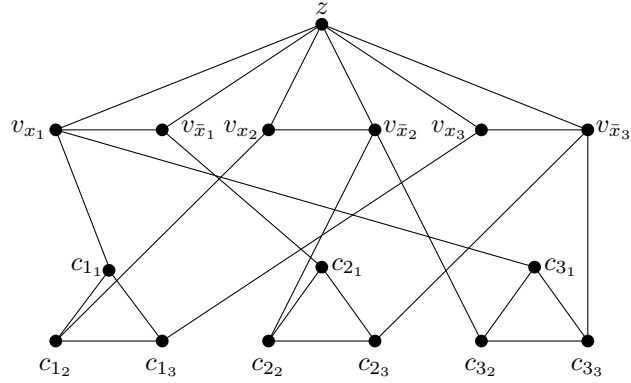
\begin{figure}[h]
\centering
\begin{tikzpicture}[scale=1.4]
	\draw (0,0)node[bnode]{}--(1,0)node[bnode]{}--(0.5,0.67)node[bnode]{}--(0,0);	
	\draw[shift={(2,0)}] (0,0)node[bnode]{}--(1,0)node[bnode]{}--(0.5,0.7)node[bnode]{}--(0,0);	
	\draw[shift={(4,0)}] (0,0)node[bnode]{}--(1,0)node[bnode]{}--(0.5,0.7)node[bnode]{}--(0,0);	
	\draw (2.5,3)--(0,2)node[bnode]{}--(1,2)node[bnode]{}--(2.5,3);
	\draw (2.5,3)--((2,2)node[bnode]{}--(3,2)node[bnode]{}--(2.5,3);
	\draw (2.5,3)--((4,2)node[bnode]{}--(5,2)node[bnode]{}--(2.5,3);
	\draw (2.5,3) node[bnode]{};
	\draw (4.5,.7)--(0,2)--(0.5,0.7);
	\draw (2.5,.7)--(1,2);
	\draw (0,0)--(2,2);
	\draw (2,0)--(3,2)--(4,0);
	\draw (1,0)--(4,2);
	\draw (3,0)--(5,2)--(5,0);
	\draw (0,-.25) node{ $c_{1_2}$};
	\draw (1,-.25) node{ $c_{1_3}$};
	\draw (2,-.25) node{ $c_{2_2}$};
	\draw (3,-.25) node{ $c_{2_3}$};
	\draw (4,-.25) node{ $c_{3_2}$};
	\draw (5,-.25) node{ $c_{3_3}$};
	\draw (0.3,0.7) node{ $c_{1_1}$};
	\draw (2.75,0.7) node{ $c_{2_1}$};
	\draw (4.75,0.7) node{ $c_{3_1}$};
	\draw (-0.25,2) node{ $v_{x_1}$};
	\draw (1.35,2) node{ $v_{\bar x_1}$};
	\draw (1.75,2) node{ $v_{x_2}$};
	\draw (3.25,2) node{ $v_{\bar x_2}$};
	\draw (3.7,2) node{ $v_{x_3}$};
	\draw (5.25,2) node{ $v_{\bar x_3}$};
	\draw (2.5,3.15) node{ $z$};
\end{tikzpicture}

\caption{The reduction from Variant A of {\sc Not-all-Equal $3$-Satisfiability} to {\sc  Independent Odd Cycle Transversal} on the instance 
$\phi=(x_1,{x}_2,x_3),(\bar{x}_1,\bar{x}_2,\bar{x}_3),(x_1,\bar{x}_2,\bar{x}_3)$. }\label{fig:3-col}
\end{figure}

We now show that $G$ has diameter at most~$4$. First note that any literal vertex is adjacent to $z$ and any clause vertex is adjacent to some literal vertex, so any vertex is at distance at most~$2$ from $z$. Therefore any two vertices are at distance at most~$4$.
	
Next, we show that $G$ is $K_{1,4}^3$-free. Any literal vertex has degree at most~$4$ since it appears in at most two clauses. However it has at most three independent neighbours since two of its neighbours, its negation and $z$, are adjacent. Each clause vertex has at most three neighbours. So the only vertex with four independent neighbours is $z$. The longest induced path with end-vertex $z$ has order at most~$4$ since any such path contains at most one literal and at most two vertices of any $T_i$. Therefore $G$ is $K_{1,4}^3$-free.

Finally, we prove that $G$ has an independent odd cycle transversal of size $m+1$ if and only if $\phi$ is not-all-equal satisfiable.

First suppose that $G$ has an independent odd cycle transversal $S$ of size $m+1$. Then $G-S$ is bipartite, say with partition classes $A$ and $B$, whereas $S$ is an independent set. Hence $G$ has a $3$-colouring with colour classes $A$, $B$ and $S$. Assume that $z$ is assigned colour~$1$. Then each literal vertex is assigned either colour $2$ or colour $3$. If, for some clause~$C_i$, the vertices corresponding to the literals of $C_i$ are all assigned the same colour, then $T_i$ cannot be coloured. Therefore, if we set literals whose vertices are coloured with colour $2$ to be true and those coloured with colour $3$ to be false, each clause must contain at least one true literal and at least one false literal.

Now suppose that $\phi$ is not-all-equal satisfiable. Then we can colour vertex $z$ with colour~$1$, each true literal with colour~$2$ and each false literal with colour~$3$. Since each clause has at least one true literal and at least one false literal, each triangle has neighbours in two different colours. This implies that each triangle is $3$-colourable. Hence, $G$ is $3$-colourable, and in particular we coloured a set $S$ of 
$m+1$ vertices of $G$ with colour~$1$ (namely, one vertex of each of the $m$ triangles $T_i$ and the vertex~$z$).
Consequently, $S$ is an independent odd cycle transversal of size~$m+1$. 
\end{proof}
	
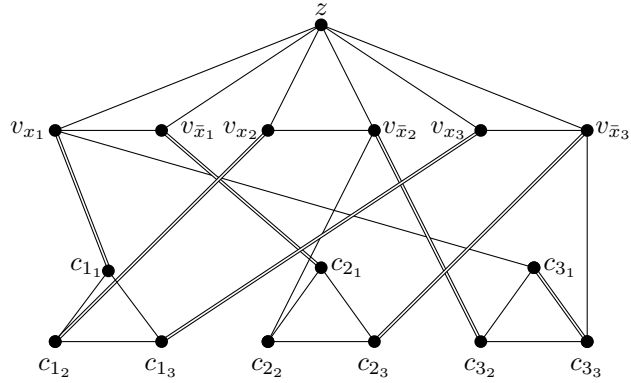
\begin{figure}
\centering
\begin{tikzpicture}[scale=1.4]
	\draw (4.5,.7)--(0,2);
	\draw[double] (0.5,0.7)--(0,2);
	\draw[double] (2.5,.7)--(1,2);
	\draw[double] (0,0)--(2,2);
	\draw (2,0)--(3,2);
	\draw[double] (3,2)--(4,0);
	\draw[double] (1,0)--(4,2);
	\draw[double] (3,0)--(5,2);
	\draw (5,2)--(5,0);
	\draw (0,-.25) node{ $c_{1_2}$};
	\draw (1,-.25) node{ $c_{1_3}$};
	\draw (2,-.25) node{ $c_{2_2}$};
	\draw (3,-.25) node{ $c_{2_3}$};
	\draw (4,-.25) node{ $c_{3_2}$};
	\draw (5,-.25) node{ $c_{3_3}$};
	\draw (0.3,0.7) node{ $c_{1_1}$};
	\draw (2.75,0.7) node{ $c_{2_1}$};
	\draw (4.75,0.7) node{ $c_{3_1}$};
	\draw (-0.25,2) node{ $v_{x_1}$};
	\draw (1.35,2) node{ $v_{\bar x_1}$};
	\draw (1.75,2) node{ $v_{x_2}$};
	\draw (3.25,2) node{ $v_{\bar x_2}$};
	\draw (3.7,2) node{ $v_{x_3}$};
	\draw (5.25,2) node{ $v_{\bar x_3}$};
	\draw (2.5,3.15) node{ $z$};
	\draw (0,0)node[bnode]{}--(1,0)node[bnode]{}--(0.5,0.67)node[bnode]{}--(0,0)node[bnode]{};	
	\draw[shift={(2,0)}] (0,0)node[bnode]{}--(1,0)node[bnode]{}--(0.5,0.7)node[bnode]{}--(0,0)node[bnode]{};	
	\draw[shift={(4,0)}] (0.5,0.7)node[bnode]{}--(0,0)node[bnode]{}--(1,0)node[bnode]{};
	\draw[double,shift={(4,0)}] (0.5,0.7)node[bnode]{}--(1,0)node[bnode]{};	
	\draw (2.5,3)--(0,2)node[bnode]{}--(1,2)node[bnode]{}--(2.5,3);
	\draw (2.5,3)--((2,2)node[bnode]{}--(3,2)node[bnode]{}--(2.5,3);
	\draw (2.5,3)--((4,2)node[bnode]{}--(5,2)node[bnode]{}--(2.5,3);
	\draw (2.5,3) node[bnode]{};
\end{tikzpicture}
	\caption{The reduction from Variant A of {\sc Not-all-Equal $3$-Satisfiability} to {\sc Acyclic $3$-Colouring} on the instance 
$\phi=(x_1,x_2,x_3),(\bar{x}_1,\bar{x}_2,\bar{x}_3),(x_1,\bar{x}_2,\bar{x}_3)$ given a covering collection 
$\mathcal{C}=\{
	(x_1,(x_1,x_2,x_3))$, $
	(\bar x_1,(\bar{x}_1,\bar{x}_2,\bar{x}_3))$, $
	(x_2,(x_1,x_2,x_3))$, $
	(\bar x_2,(x_1,\bar{x}_2,\bar{x}_3))$, $
	(x_3,(x_1,x_2,x_3))$, $
	(\bar x_3,(\bar{x}_1,\bar{x}_2,\bar{x}_3))\}$. 
Double lines are edges that are substitution instances of $K_{2,3}$ where the endpoints come from the partition of size~$2$.}
	\label{fig:acyclic-3-col-polyad}
\end{figure}

\begin{theorem}\label{t-hard2}
{\sc Acyclic $3$-Colouring} is \NP-complete for $K_{1,6}^5$-free graphs of diameter~$6$.
\end{theorem}

\begin{proof}	
	Suppose now we attempt to frame a similar argument for  {\sc Acyclic $3$-Colouring} by putting the construction used in the proof of Theorem~\ref{t-hard1} through a process by which we map edges to bipartite graphs $K_{2,3}$ (which is the standard reduction from {\sc $3$-Colouring} to {\sc Acyclic $3$-Colouring}). That is, for {\sc Acyclic $3$-Colouring}, we replace each edge $u_1u_2$ by three new vertices $w_1,w_2,w_3$ and edges $u_iw_j$ for $i\in \{1,2\}$ and $j\in \{1,2,3\}$.
Alas, using vertex $z$, which is the only vertex of unbounded degree, we can find arbitrarily large polyads. 
However, we reduce from variant A and only substitute the edges of a certain matching. 
	
Let $\phi$ be an instance of variant A and $\mathcal{C}$ be a given covering collection of $\phi$. We only substitute edges $w_1w_2$ in the construction used in the proof of Theorem~\ref{t-hard1} for which 

\begin{itemize}
\item $w_1w_2=v_{z_1}c_{i_j}$ such that $(z_1,C_i)\in {\cal C}$; or
\item $w_1w_2=c_{i_h}c_{i_j}$ such that
 there are two literals $z_1$ and $z_2$ with $(z_1,C_i),(z_2,C_i)\notin\mathcal{C}$, and $c_{i_h}$ and $c_{i_j}$
 are the two vertices of $T_i$ adjacent to $v_{z_1}$ and $v_{z_2}$, respectively.
\end{itemize} 

\noindent This is drawn in Figure~\ref{fig:acyclic-3-col-polyad}. We note that the set of edges we substitute is indeed a matching.
We claim that the resulting graph $G$ has diameter at most $6$. Indeed, there is path of length at most $3$ from any vertex to $z$. We further claim that $G$ is $K_{1,6}^{5}$-free, which follows from the fact that any induced $K_{1,6}^{5}$ in the graph would have to involve the vertex $z$ as its centre 
(as the other vertices have at most~$5$ independent neighbours)
 and all paths of order $7$ starting from $z$ must involve two vertices adjacent to $z$.

Finally, we claim that $G$ is acyclically $3$-colourable if and only if $\phi$ is not-all-equal satisfiable.
	
First suppose that $G$ is acyclically $3$-colourable. 	Given an acyclic $3$-colouring of $G$, assume $z$ is assigned colour $1$. Then each literal vertex is assigned either colour~$2$ or colour~$3$. The argument here concludes as it does in the reduction to {\sc $3$-Colouring}, since in any acyclic $3$-colouring of the gadget $K_{2,3}$ the vertices in the partition of size~$2$ must receive distinct colours.
	
Now suppose that $\phi$ is not-all-equal satisfiable. Then, we can colour vertex $z$ with colour~$1$, each true literal $v_{x_i}$ with colour $2$ and each false literal $v_{x_i}$ with colour~$3$. The proof concludes as in the reduction to {\sc $3$-Colouring} except that we must argue that there are no bichromatic cycles.
For a contradiction, assume that there is a bichromatic cycle, say $C$.
Bearing in mind that in a substitution instance of $K_{2,3}$, the vertices of the partition of size~$2$ are coloured distinct, we note that $C$ cannot contain any edge that belongs to such a gadget.
As $C$ is a bichromatic cycle, we find that it contains at least two adjacent vertices $v_{z_1}$ and $c_{i_{j_1}}$.
Moreover, $C$ contains exactly one other vertex, say $c_{i_{j_2}}$, from $\{c_{i_1},c_{i_2},c_{i_3}\}$, and this vertex is adjacent to $c_{i_{j_1}}.$ 
It follows that some $v_{z_2}$ for $z_2\neq z_1$ is the second neighbour of $c_{i_{j_2}}$ in $C$.
But now $(z_1,C_i),(z_2,C_i)\notin \mathcal{C}$ and $c_{i_{j_1}}c_{i_{j_2}}$ is an edge of $G$, which contradicts our substitution rule.
\end{proof}

\begin{theorem}\label{t-hard3}
{\sc Star $3$-Colouring} is \NP-complete for $K_{1,6}^{14}$-free graphs of diameter~$14$.
\end{theorem}

\begin{proof}	
	Our construction for {\sc Acyclic $3$-Colouring} in the proof of Theorem~\ref{t-hard3} cannot be directly adapted for {\sc Star $3$-Colouring}. Instead of this, we reduce from Variant~B of {\sc Not-All-Equal $3$-Satisfiability}. Hence, we are given a CNF formula $\phi$ with clauses $C_1,\ldots,C_m$ and variables $x_1,\ldots,x_n$ such that 
each $C_i$ consists of three positive literals and each  literal  occurs  in  at  most  four  different  clauses. From~$\phi$, we construct a graph $G$ as follows (see also Figure~\ref{fig:star-3-col-polyad}):
	
	\begin{itemize}
		\item Add a vertex $v_{x_i}$ for each variable $x_i$.
		\item Add a vertex $z$ adjacent to each vertex $v_{x_i}$.
		\item Add two new vertices $z',z''$ in a triangle with $z$.
		\item Add vertices $p^1_{x_i},p^2_{x_i},p^3_{x_i},p^4_{x_i}$ for each instance of a variable $x_i$ with edges from each of these to $v_{x_i}$.
		\item Add vertices $q^1_{x_i},q^2_{x_i},q^3_{x_i},q^4_{x_i}$ for each instance of a variable $x_i$ with edges from each $q^j_{x_i}$ to $p^j_{x_i}$ which are substitution instances of $K_{2,2}$ and the endpoints come from the same part.
		\item For each clause $C_i$ add a triangle $T_i$ with vertices $c_{i_1}, c_{i_2}, c_{i_3}$ where edges are substitution instance of $K_{2,2}$ and the endpoints come from the same part.
		\item 
		Fix an arbitrary order of the literals of every $C_i$, $x_{i_1}, x_{i_2}, x_{i_3}$.
 Assign every pair $(i,j)$ a vertex of
$q_{x_{i_j}}^1,q_{x_{i_j}}^2,q_{x_{i_j}}^3$ and make this vertex adjacent to
$c_{i_j}$, such that this assignment is injective. Let each of the new edges be a substitution instance of $K_{2,2}$, where the
endpoints come from the same part.
	\end{itemize}

\noindent
	We draw the special edges that are in fact built from instances $K_{2,2}$ with double lines in Figure~\ref{fig:star-3-col-polyad}. We claim that $G$ has diameter at most $14$. Indeed, every vertex has a path of length at most $7$ from it to $z$. We also claim that $G$ is $K_{1,6}^{14}$-free. Vertices that are not equal to some $c_{i_j}$ or $z$ have degree at most~$5$ so cannot be the centre of an induced $K_{1,6}^{14}$. Moreover, every $c_{i_j}$ 
has degree~$6$, but for every $c_{i_j}$ there is no vertex that is adjacent to exactly one of the neighbours of $c_{i_j}$.  Hence, no $c_{i_j}$ can be the centre of an induced $K_{1,6}^{14}$ either. Finally, $z$ cannot be the centre of an induced $K_{1,6}^{14}$ either, as every induced path that starts from $z$ has length at most~$14$. Hence, $G$ is indeed $K_{1,6}^{14}$-free.
	
\begin{figure}[h]
\centering
\begin{tikzpicture}[scale=1.4]
   \draw (2.5,6) node[bnode]{}--(2,6.7) node[bnode]{}--(3,6.7)
node[bnode]{}--(2.5,6);
        \draw (2,6.9) node{ $z'$};
        \draw (3,6.9) node{ $z''$};
	\draw (0,.75) node{ $c_{1_1}$};
	\draw (1,.75) node{ $c_{1_3}$};
	\draw (2,.75) node{ $c_{2_1}$};
	\draw (3,.75) node{ $c_{2_3}$};
	\draw (4,.75) node{ $c_{3_1}$};
	\draw (5,.75) node{ $c_{3_3}$};
	\draw (0.25,1.7) node{ $c_{1_2}$};
	\draw (2.25,1.7) node{ $c_{2_2}$};
	\draw (4.25,1.7) node{ $c_{3_2}$};
	\draw (-1.45,3) node{$q_{x_1}^1$};
	\draw (-0.95,3) node{$q_{x_1}^2$};
	\draw (-.45,3) node{$q_{x_1}^3$};
	\draw (0.05,3) node{$q_{x_1}^4$};
	\draw (-1.45,4) node{$p_{x_1}^1$};
	\draw (-.95,4) node{$p_{x_1}^2$};
	\draw (-.45,4) node{$p_{x_1}^3$};
	\draw (0.05,4) node{$p_{x_1}^4$};
	\draw (-0.75,5) node{$v_{x_1}$};
	\draw (1.25,5) node{$v_{x_2}$};
	\draw (3.75,5) node{$v_{x_3}$};
	\draw (5.75,5) node{$v_{x_4}$};
	\draw (2.5,6.15) node{ $z$};
	\draw[double,shift={(0,1)}] (0,0)node[bnode]{}--(1,0)node[bnode]{}--(0.5,0.7)node[bnode]{}--(0,0)node[bnode]{};	
	\draw[double,shift={(2,1)}] (0,0)node[bnode]{}--(1,0)node[bnode]{}--(0.5,0.7)node[bnode]{}--(0,0)node[bnode]{};	
	\draw[double,shift={(4,1)}] (0,0)node[bnode]{}--(1,0)node[bnode]{}--(0.5,0.7)node[bnode]{}--(0,0)node[bnode]{};	
	\draw (2.5,6)--(-0.5,5)node[bnode]{}--(-1.25,4)node[bnode]{};
	\draw (-0.5,5)node[bnode]{}--(-.75,4)node[bnode]{};
	\draw (-0.5,5)node[bnode]{}--(-.25,4)node[bnode]{};
	\draw (-0.5,5)node[bnode]{}--(.25,4)node[bnode]{};
	\draw (2.5,6)--(1.5,5)node[bnode]{}--(0.75,4)node[bnode]{};
	\draw (1.5,5)node[bnode]{}--(1.25,4)node[bnode]{};
	\draw (1.5,5)node[bnode]{}--(1.75,4)node[bnode]{};
	\draw (1.5,5)node[bnode]{}--(2.25,4)node[bnode]{};
	\draw (2.5,6)--(3.5,5)node[bnode]{}--(2.75,4)node[bnode]{};
	\draw (3.5,5)node[bnode]{}--(3.25,4)node[bnode]{};
	\draw (3.5,5)node[bnode]{}--(3.75,4)node[bnode]{};
	\draw (3.5,5)node[bnode]{}--(4.25,4)node[bnode]{};
	\draw (2.5,6)--(5.5,5)node[bnode]{}--(4.75,4)node[bnode]{};
	\draw (5.5,5)node[bnode]{}--(5.25,4)node[bnode]{};
	\draw (5.5,5)node[bnode]{}--(5.75,4)node[bnode]{};
	\draw (5.5,5)node[bnode]{}--(6.25,4)node[bnode]{};
	\draw[double] (-1.25,4)node[bnode]{}--(-1.25,3)node[bnode]{}--(0,1)node[bnode]{};
	\draw[double] (-.75,4)node[bnode]{}--(-.75,3)node[bnode]{}--(2,1)node[bnode]{};
	\draw[double] (-.25,4)node[bnode]{}--(-.25,3)node[bnode]{};
	\draw[double] (.25,4)node[bnode]{}--(.25,3)node[bnode]{};
	\draw[double] (.75,4)node[bnode]{}--(.75,3)node[bnode]{}--(0.5,1.7)node[bnode]{};
	\draw[double] (1.25,4)node[bnode]{}--(1.25,3)node[bnode]{}--(4,1)node[bnode]{};
	\draw[double] (1.75,4)node[bnode]{}--(1.75,3)node[bnode]{};
	\draw[double] (2.25,4)node[bnode]{}--(2.25,3)node[bnode]{};
	\draw[double] (2.75,4)node[bnode]{}--(2.75,3)node[bnode]{}--(1,1)node[bnode]{};
	\draw[double] (3.25,4)node[bnode]{}--(3.25,3)node[bnode]{}--(2.5,1.7)node[bnode]{};
	\draw[double] (3.75,4)node[bnode]{}--(3.75,3)node[bnode]{}--(4.5,1.7)node[bnode]{};
	\draw[double] (4.25,4)node[bnode]{}--(4.25,3)node[bnode]{};
	\draw[double] (4.75,4)node[bnode]{}--(4.75,3)node[bnode]{}--(3,1)node[bnode]{};
	\draw[double] (5.25,4)node[bnode]{}--(5.25,3)node[bnode]{}--(5,1)node[bnode]{};
	\draw[double] (5.75,4)node[bnode]{}--(5.75,3)node[bnode]{};
	\draw[double] (6.25,4)node[bnode]{}--(6.25,3)node[bnode]{};
	\draw (2.5,6) node[bnode]{};
\end{tikzpicture}
	\caption{The reduction from variant~B of {\sc Not-All-Equal 3-Satisfiability} to {\sc Star $3$-Colouring} on the instance 
$\phi=(x_1,x_2,x_3),(x_1,x_3,x_4),(x_2,x_3,x_4)$. Double lines are edges that are substitution instances of $K_{2,2}$ where the endpoints come from same partition.}
	\label{fig:star-3-col-polyad}
\end{figure}
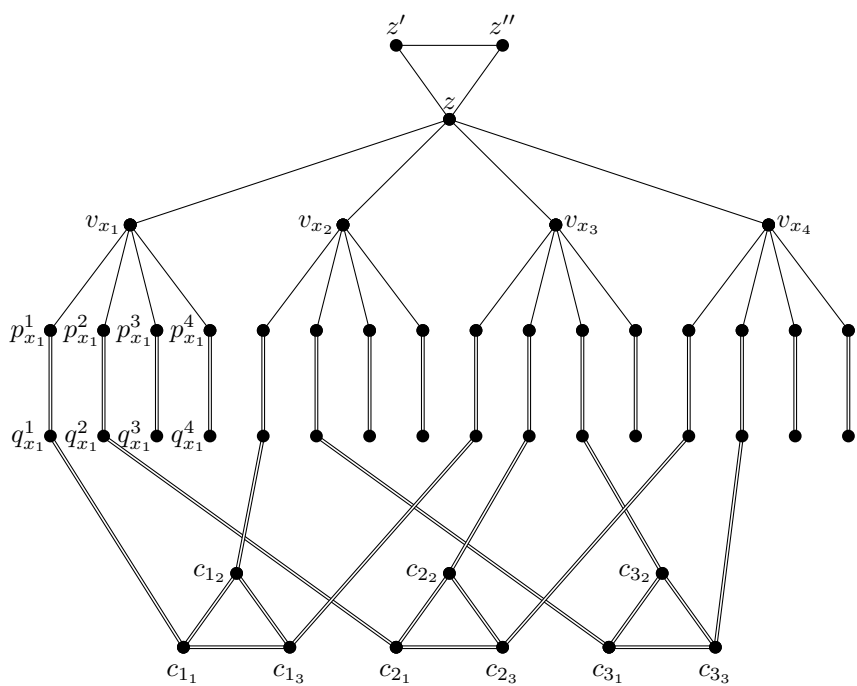

Finally, we claim that $G$ is star $3$-colourable if and only if $\phi$ is not-all-equal satisfiable.
			First suppose that $G$ is star $3$-colourable. Given a star $3$-colouring of $G$, assume $z$ is assigned colour $1$. Then each vertex $v_{x_i}$ is assigned either colour $2$ or colour $3$. Now each of the vertices $p^1_{x_i},p^2_{x_i},p^3_{x_i}$ is assigned the same colour from $\{2,3\}$ (this is enforced by the fact that $\{z',z''\}$ must be coloured $\{2,3\}$ which forbids any possibility that colour~$1$ is used). Furthermore each of the vertices $q^1_{x_i},q^2_{x_i},q^3_{x_i}$ is assigned precisely the colour from $\{2,3\}$ that $p^1_{x_i},p^2_{x_i},p^3_{x_i}$ was not assigned. The argument here concludes as it does in the reduction to {\sc $3$-Colouring}.
	
Now suppose that $\phi$ is not-all-equal satisfiable, then we can colour vertex $z$ with colour~$1$, the vertices $q^1_{x_i},q^2_{x_i},q^3_{x_i}$ of each true literal with colour $2$ and the vertices $q^1_{x_i},q^2_{x_i},q^3_{x_i}$ of each false literal with colour $3$. Then, since each clause has at least one true literal and at least one false literal, each triangle has neighbours in two different colours. This implies that each triangle is $3$-colourable. The argument here concludes as it does in the reduction to {\sc $3$-Colouring}.
\end{proof}

\section{Conclusions}\label{s-con}

We showed that bounding the diameter does not help for {\sc Independent Set} for $H$-free graphs. We proved that this does help for some problems related to {\sc $3$-Colouring} if $H$ is the chair.  Whether these results can be extended to larger polyads~$H$ is an interesting but challenging task. For three of these problems we gave a polyad~$H$ such that they are \NP-complete for $H$-free graphs of diameter~$d$ for some constant $d$.
Such a polyad $H$ was already known to exist for {\sc $3$-Colouring} of graphs of diameter at most~$4$~\cite{MPS19}.
We ask if there exists a polyad $H$ and an integer~$d$ such that {\sc Near-Bipartiteness} and {\sc Independent Feedback Vertex Set} are \NP-complete for $H$-free graphs of diameter at most~$d$.


\begin{thebibliography}{10}

\bibitem{Al82}
Vladimir~E. Alekseev.
\newblock The effect of local constraints on the complexity of determination of
  the graph independence number.
\newblock {\em Combinatorial-Algebraic Methods in Applied Mathematics}, pages
  3--13, 1982 (in Russian).

\bibitem{Al04}
Vladimir~E. Alekseev.
\newblock Polynomial algorithm for finding the largest independent sets in
  graphs without forks.
\newblock {\em Discrete Applied Mathematics}, 135:3--16, 2004.

\bibitem{BKM12}
Manuel Bodirsky, Jan K{\'{a}}ra, and Barnaby Martin.
\newblock The complexity of surjective homomorphism problems - a survey.
\newblock {\em Discrete Applied Mathematics}, 160:1680--1690, 2012.

\bibitem{BJMOPS}
Jan Bok, Nikola Jedli\u{c}kov\'{a}, Barnaby Martin, Dani{\"{e}}l Paulusma, and
  Pascal Ochem~Siani Smith.
\newblock Acyclic, star and injective colouring: A~complexity picture for
  ${H}$-free graphs.
\newblock {\em CoRR}, abs/2008.09415 (conference version in Proc. {ESA} 2020,
  LIPIcs 173, 22:1--22:22), 2021.

\bibitem{BDFJP18}
Marthe Bonamy, Konrad~K. Dabrowski, Carl Feghali, Matthew Johnson, and
  Dani{\"{e}}l Paulusma.
\newblock Independent feedback vertex sets for graphs of bounded diameter.
\newblock {\em Information Processing Letters}, 131:26--32, 2018.

\bibitem{BDFJP19}
Marthe Bonamy, Konrad~K. Dabrowski, Carl Feghali, Matthew Johnson, and Dani\"el
  Paulusma.
\newblock Independent feedback vertex set for ${P}_5$-free graphs.
\newblock {\em Algorithmica}, 81:1342--1369, 2019.

\bibitem{BLS99}
Andreas Brandst{\"a}dt, Van~Bang Le, and Jeremy~P. Spinrad.
\newblock {\em Graph Classes: A Survey}, volume~3 of {\em SIAM Monographs on
  Discrete Mathematics and Applications}.
\newblock SIAM, 1999.

\bibitem{BGMPS21}
Christoph Brause, Petr~A. Golovach, Barnaby Martin, Dani\"el Paulusma, and
  Siani Smith.
\newblock Acyclic, star and injective colouring: bounding the diameter.
\newblock {\em Proc. WG 2021, LNCS}, 12911:336--348, 2021.

\bibitem{BFGP13}
Hajo Broersma, Fedor~V. Fomin, Petr~A. Golovach, and Dani\"el Paulusma.
\newblock Three complexity results on coloring {$P_k$}-free graphs.
\newblock {\em European Journal of Combinatorics}, 34(3):609--619, 2013.

\bibitem{CHJMP18}
Nina Chiarelli, Tatiana~R. Hartinger, Matthew Johnson, Martin Milani\v{c}, and
  Dani\"el Paulusma.
\newblock Minimum connected transversals in graphs: new hardness results and
  tractable cases using the price of connectivity.
\newblock {\em Theoretical Computer Science}, 705:75--83, 2018.

\bibitem{Ch12}
Maria Chudnovsky.
\newblock The structure of bull-free graphs {II} and {III} - {A} summary.
\newblock {\em Journal of Combinatorial Theory, Series {B}}, 102:252--282,
  2012.

\bibitem{CS05}
Maria Chudnovsky and Paul~D. Seymour.
\newblock The structure of claw-free graphs.
\newblock {\em Surveys in Combinatorics, London Mathematical Society Lecture
  Note Series}, 327:153--171, 2005.

\bibitem{DJP19}
Konrad~K. Dabrowski, Matthew Johnson, and Dani\"el Paulusma.
\newblock Clique-width for hereditary graph classes.
\newblock {\em Proc. BCC 2019, London Mathematical Society Lecture Note
  Series}, 456:1--56, 2019.

\bibitem{DD20}
Andreas Darmann and Janosch D{\"{o}}cker.
\newblock On a simple hard variant of not-all-equal 3-sat.
\newblock {\em Theoretical Computer Science}, 815:147--152, 2020.

\bibitem{DPR}
Micha{\l} D{\k e}bski, Marta Piecyk, and Pawe{\l} Rz{\k a}\.{z}ewski.
\newblock Faster $3$-coloring of small-diameter graphs.
\newblock {\em Proc. ESA 2021, LIPIcs}, 204:37:1--37:15, 2021.

\bibitem{Ed86}
Keith Edwards.
\newblock The complexity of colouring problems on dense graphs.
\newblock {\em {Theoretical Computer Science}}, 43:337--343, 1986.

\bibitem{GJPS17}
Petr~A. Golovach, Matthew Johnson, Dani\"el Paulusma, and Jian Song.
\newblock A survey on the computational complexity of colouring graphs with
  forbidden subgraphs.
\newblock {\em Journal of Graph Theory}, 84:331--363, 2017.

\bibitem{HMLW11}
Danny Hermelin, Matthias Mnich, Erik~Jan van Leeuwen, and Gerhard~J. Woeginger.
\newblock Domination when the stars are out.
\newblock {\em ACM Transactions on Algorithms}, 15:25:1--25:90, 2019.

\bibitem{HoTa73}
John Hopcroft and Robert Tarjan.
\newblock Algorithm 447: Efficient algorithms for graph manipulation.
\newblock {\em Communications of the ACM}, 16(6):372--378, 1973.

\bibitem{KKTW01}
Daniel Kr{\'a}l', Jan Kratochv\'{\i}l, {\relax Zs}olt Tuza, and Gerhard~J.
  Woeginger.
\newblock Complexity of coloring graphs without forbidden induced subgraphs.
\newblock {\em Proc. WG 2001, LNCS}, 2204:254--262, 2001.

\bibitem{MPS19}
Barnaby Martin, Dani{\"{e}}l Paulusma, and Siani Smith.
\newblock Colouring ${H}$-free graphs of bounded diameter.
\newblock {\em Proc. MFCS 2019, LIPIcs}, 138:14:1--14:14, 2019.

\bibitem{MPS21}
Barnaby Martin, Dani{\"{e}}l Paulusma, and Siani Smith.
\newblock Colouring graphs of bounded diameter in the absence of small cycles.
\newblock {\em Discrete Appleid Mathematics}, 314:150--161, 2022.

\bibitem{MS16}
George~B. Mertzios and Paul~G. Spirakis.
\newblock Algorithms and almost tight results for $3$-{C}olorability of small
  diameter graphs.
\newblock {\em Algorithmica}, 74:385--414, 2016.

\bibitem{Mu17}
Andrea Munaro.
\newblock On line graphs of subcubic triangle-free graphs.
\newblock {\em Discrete Mathematics}, 340:1210--1226, 2017.

\bibitem{PPR}
Giacomo Paesani, Dani\"el Paulusma, and Pawe{\l} Rz{\k a}\.{z}ewski.
\newblock {Feedback Vertex Set and Even Cycle Transversal for ${H}$-free
  graphs: finding large block graphs}.
\newblock {\em Proc. MFCS 2021, LIPIcs}, 202:82:1--82:14, 2021.

\bibitem{Pa15}
Dani\"el Paulusma.
\newblock Open problems on graph coloring for special graph classes.
\newblock {\em Proc. WG 2015, LNCS}, 9224:16--30, 2015.

\bibitem{Po74}
Svatopluk Poljak.
\newblock A note on stable sets and colorings of graphs.
\newblock {\em Commentationes Mathematicae Universitatis Carolinae},
  15:307--309, 1974.

\bibitem{RS04}
Bert Randerath and Ingo Schiermeyer.
\newblock Vertex colouring and forbidden subgraphs - {A} survey.
\newblock {\em Graphs and Combinatorics}, 20:1--40, 2004.

\bibitem{Sb80}
Najiba {Sbihi}.
\newblock {Algorithme de recherche d'un stable de cardinalit\'e maximum dans un
  graphe sans \'etoil\'e}.
\newblock {\em {Discrete Mathematics}}, 29:53--76, 1980.

\bibitem{Sc78}
Thomas~J. Schaefer.
\newblock The complexity of satisfiability problems.
\newblock {\em Proc. STOC 1978}, pages 216--226, 1978.

\end{thebibliography}
\end{document}